\documentclass[10pt,journal,compsoc]{IEEEtran}

\usepackage{graphicx}
\usepackage{balance}
\usepackage{color}
\usepackage[noend]{algpseudocode}
\usepackage{algorithm}
\usepackage{varwidth}
\usepackage{url}
\usepackage{multirow}
\usepackage{subfigure}
\usepackage{mathtools}
\usepackage{amsmath,bm,amssymb}
\usepackage{hyperref}
\usepackage{times}
\usepackage{diagbox}
\usepackage{pifont}
\usepackage{indentfirst}
\usepackage{verbatim}
\usepackage{cite}
\usepackage{amsthm}
\usepackage[numbers,sort&compress]{natbib}

\newcommand{\tabincell}[2]{\begin{tabular}{@{}#1@{}}#2\end{tabular}}

\newtheorem{definition}{Definition}

\newtheorem{problem}{Problem}
\newtheorem{example}{Example}
\newtheorem{lemma}{Lemma}
\newtheorem{proposition}{Proposition}
\newtheorem{property}{Property}
\begin{document}

\title{Exploring Communities in Large Profiled Graphs}

\author{Yankai Chen, Yixiang Fang, Reynold Cheng~\IEEEmembership{Member,~IEEE}, Yun Li, Xiaojun Chen, Jie Zhang
\IEEEcompsocitemizethanks{
\IEEEcompsocthanksitem Y. Chen, Y. Fang, and R. Cheng are with the Department
of Computer Science, The University of Hong Kong, Hong Kong.
\protect\\
E-mail: \{ykchen, yxfang, ckcheng\}@cs.hku.hk
\IEEEcompsocthanksitem Y. Li is with Department of Computer Science and Technology, Nanjing University, China.
\protect\\
E-mail: liycser@gmail.com
\IEEEcompsocthanksitem X. Chen is with College of Computer Science and Software,Shenzhen University, China.
\protect\\
E-mail: xjchen@szu.edu.cn
\IEEEcompsocthanksitem J. Zhang is with School of Computer Science and Engineering, Nanyang Technological University, Singapore.  
\protect\\
E-mail: zhangj@ntu.edu.sg
}
\thanks{Manuscript received March 20, 2018.}}


\IEEEtitleabstractindextext{%

\begin{abstract}

Given a graph $G$ and a vertex $q\in G$, the community search (CS) problem aims to efficiently find a subgraph of $G$ whose vertices are closely related to $q$.  Communities are prevalent in social and biological networks, and can be used in product advertisement and social event recommendation. In this paper, we study \emph{profiled community search} (PCS), where CS is performed on a \emph{profiled graph}.  This is a graph in which each vertex has labels arranged in a hierarchical manner.  Extensive experiments show that PCS can identify communities with themes that are common to their vertices, and is more effective than existing CS approaches.  As a naive solution for PCS is highly expensive, we have also developed a tree index, which facilitate efficient and online solutions for PCS.

\end{abstract}

}
\maketitle

\section{Introduction}
\label{sec:intro}

\IEEEPARstart{D}ue to the recent developments of gigantic social networks (e.g., Flickr, Facebook, and Twitter), topics of graph queries have attracted attention from industry and research areas~\cite{ding2007finding,fang2014detecting,he2007blinks,kacholia2005bidirectional,kargar2011keyword,xu2012model,yu2009keyword}. 
Communities, which are often found in large graphs, can be used in various applications, such as social event setting, friend recommendation, and research collaboration analysis~\cite{sozio2010community,online-sigmod2013,k-truss2014,fang2016effective,huang2017attribute}. Given a graph $G$ and a query vertex $q\in G$, the goal of \emph{community search} (CS) is to extract \emph{communities}, or densely connected subgraphs of $G$ that contain $q$, in an online manner. 

\begin{figure}[htp]
    \centering
        \subfigure[a profiled graph]{
            \includegraphics[width=0.8\columnwidth]{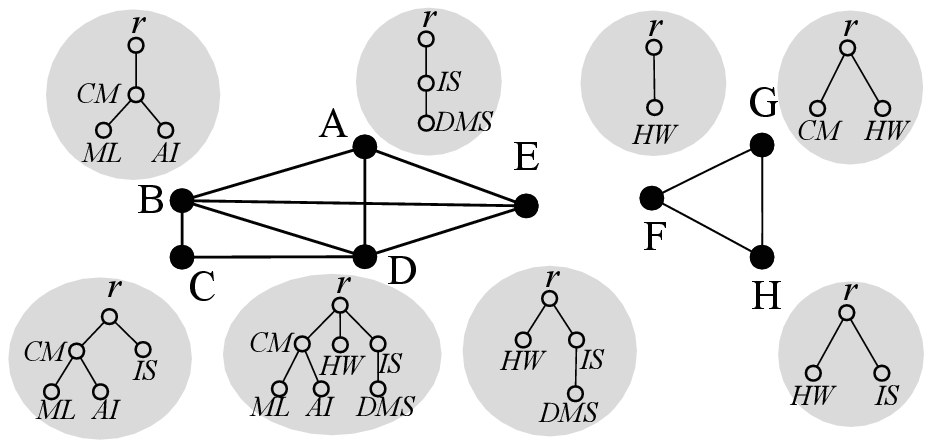}
            \label{fig:pg}
        }
       
        \subfigure{
            \includegraphics[width=0.95\columnwidth]{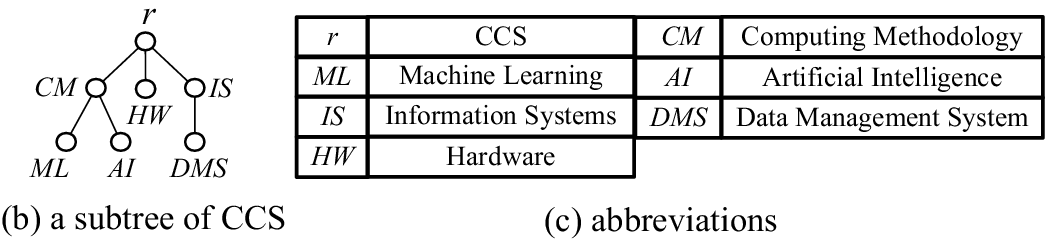}
            \label{fig:abbreviations}
        }
    \caption{ A profiled graph, a subtree of CCS and meanings of terms.}\label{fig:intro}
\end{figure}

In this paper, we investigate the CS problem for a \emph{profiled graph}. This is essentially a kind of {\it attributed graphs}, where each graph vertex is associated with a set of labels arranged in a hierarchical manner called a {\it P-tree}. Fig.~\ref{fig:pg} shows a profiled graph, which is a computer science collaboration network; each vertex represents a researcher, and a link between two vertices depicts that the two corresponding researchers have worked together before. Each vertex is associated with a P-tree, which describes the expertise of researchers. Fig.1(c) shows the meanings of the terms in each P-tree, following the \emph{ACM Computing Classification System (CCS)}~\footnote{ACM CCS: http://www.acm.org/publications/class-2012}, which is partially presented in Fig.1(b). For instance, vertex $B$ denotes a researcher, whose research domain is in computing methodology ({\it CM}), with specific interest in machine learning ({\it ML}) and artificial intelligence ({\it AI}). Profiled graphs are informative and can be found in various graph applications (e.g., knowledge bases, social and collaboration networks). Moreover, the P-trees of profiled graphs systematically organize labels related to a vertex (e.g., hierarchical and interrelated knowledge in knowledge bases, affiliation, expertise, and locations in social and collaboration networks), reflecting the semantic relationship among them. For example, in a P-tree, label ``London'' can be a child node of ``UK'', because London is a UK city.

\begin{figure}[htp]
    \centering
        \subfigure[two PC's]{
            \includegraphics[width=0.3\columnwidth]{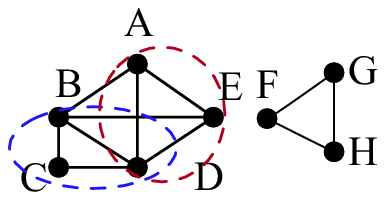}
            \label{fig:pc}{}
        }
        \hspace{0.2in}
        \subfigure[\{B, C, D\}]{
            \includegraphics[width=0.16\columnwidth]{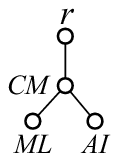}
            \label{fig:ptree1}
        }
        \hspace{0.2in}
        \subfigure[\{A, D, E\}]{
            \includegraphics[width=0.16\columnwidth]{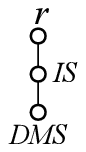}
            \label{fig:ptree2}
        }
    \caption{Illustrating profiled community search (PCS).}\label{fig:pcs}
\end{figure}

{\bf Prior works.}
The methods related to retrieval communities can generally be classified into \emph{community detection} (CD) methods and \emph{community search} (CS) methods. In general, the aim of CD algorithms is to retrieve all communities for a graph~\cite{fortunato2010community,liu2009topic,attr-topic-kdd2008,newman2004finding,attr-www2013,attr-topic-sigmod2012,yang2013community,zhou2009graph}. Note that these solutions are not “query-based”. This means that, given a user-specified query vertex, they are not customized for a query request. As a result, these algorithms normally take a long time to find all the communities for a large graph. Thus it is not suitable to use CD algorithms for quick or online retrieval of communities. 
To solve these problems, CS solutions have been recently proposed~\cite{barbieri2015efficient,cui2013online,local2014,k-truss2014,huang2015approximate,sozio2010community}. Compared with CD solutions, CS approaches are query-based, and thus are suitable to derive communities in an “online” manner. 

However, to our best knowledge, previous CS algorithms are not designed for profiled graphs. Early solutions (e.g., \cite{sozio2010community,online-sigmod2013,k-truss2014}) often only consider graph topology (e.g., a $k$-core is a community such that each vertex is connected to $k$ or more vertexes). They did not consider the use of vertex labels. As pointed out in \cite{fang2016effective}, the communities returned by those solutions are often huge (e.g., a community can easily contain over $1,000$ vertices). Moreover, the vertices included in the communities were not quite related. Recent works, such as
ACQ~\cite{fang2016effective} and ATC~\cite{huang2017attribute},
propose to use both graph structure and vertex label information.  While these works have been shown to be more effective than CS solutions that do not utilize vertex labels, they did not employ the hierarchical relationship among labels (e.g., P-trees in Fig.~\ref{fig:pg}). This may lead to suboptimal results. In Fig.~\ref{fig:pg}, 
suppose that a renowned expert $D$ wants to organize a seminar where researchers are closely related to each other. 
Based on the ACQ solution~\cite{fang2016effective}, with $k$=2, only a 2-core is searched (Fig.~\ref{fig:ptree1}), whose vertices \{$B$, $C$, $D$\} have several labels (i.e., {\it r}, {\it CM}, {\it ML}, {\it AI}) in common. However, it fails to return the community in Fig.~\ref{fig:ptree2}, whose vertices are also highly similar.
For these two communities, the shared labels as well as their relationships in the P-tree are very different. Therefore, both communities can be presented to the organizer for further selection.

{\bf Profiled community search.} In this paper, we study \emph{profiled community search} (PCS), which aims to find \emph{profiled communities}, or PC's, for a profiled graph. To obtain high-quality communities, we use \emph{structure cohesiveness} and \emph{profile cohesiveness} to constrain PC's. We adopt widely used metric \emph{minimum degree}~\cite{sozio2010community,cui2014local,li2015influential,batagelj2003m,dorogovtsev2006k,seidman1983network} to measure the structure cohesiveness. Note that in PCS problem, the minimum degree matric can be replaced by other useful matrics, e.g., $k$-truss \cite{k-truss2014} and $k$-clique~\cite{cui2013online}, to fit in other possible application scenarios. In a profiled graph, each vertex is associated with a P-tree. To measure the profile cohesiveness, we fully utilize the information in P-trees. 
Conceptually, a PC is a group of densely connected vertices, whose P-trees have the largest degree of overlap. This overlapping part is the largest common subtree shared by all the vertices. 
Fig.~\ref{fig:pc} illustrates two PC's in the profiled graph of Fig.~\ref{fig:intro}, namely \{$B$, $C$, $D$\} and \{$A$, $D$, $E$\}. In Fig.~\ref{fig:ptree1} and Fig.~\ref{fig:ptree2}, the two PC's, as well as their largest common subtrees are respectively shown.
For example, in Fig.~\ref{fig:ptree2}, vertices $A$, $D$, and $E$ all possess the subtree with root $r$ and leaf nodes {\it ``IS''} and {\it ``DMS''}. Notice that these three vertices also form a 2-core of $D$, and the common subtree among them is the largest.
The common subtree sufficiently reflects the ``theme'' of the community. In the PC of Fig.~\ref{fig:ptree1}, all the researchers involved share interest in machine learning and artificial intelligence, whereas for Fig.~\ref{fig:ptree2}, the researchers are all interested in information systems and hardware studies.

\noindent{\bf Personalization.} PCS problem allows a query user to search communities that exhibits both structure cohesiveness and profiled cohesiveness. The parameter $k$ controls the density of connection intensiveness. The profiled cohesiveness constrains the community to be semantically similar as much as possible. For instance, PCS methods can answer questions such as who are my close friends so that we have strong connection and common intresets and expertise? In contrast, existing CD methods \cite{guo2008regionalization,expert2011uncovering,chen2015finding} often use some global criteria (e.g., modularity) where the graph is partitioned a-priori with no reference to the particular query vertices. Thus existing CD methods are not suitable for personalized queries. 

\noindent{\bf Online search.} Similar to other online CS approaches, our PCS method is able to find PC's from a large-scale profiled graph effectively and efficiently. However, existing CD methods for graph query problems are generally slower. This is mainly because that they are designed for retrieving all the communities for an entire graph.

{\bf Contributions.} As we will explain, a simple solution to solve the PCS problem is extremely expensive. To improve the efficiency of finding PC's (so that they can be used in online applications), we first introduce an {\it anti-monotonicity} property, which allows the candidates for a PC to be pruned efficiently. We further develop the \emph{CP-tree} index, which systematically organizes the graph vertices and P-trees of a profiled graph. The CP-tree index enables the development of two fast PC discovery algorithms. We experimentally evaluate our solutions on two real large profiled graphs and two synthetic profiled graphs. Our results show that PC's are better representations of communities, and the CP-tree based algorithms are up to 4 order-of-magnitude faster than basic solution.

{\bf Organization.} We review the related work in Section~\ref{sec:related}. Section~\ref{sec:problem} presents the PCS problem and a basic solution. Section~\ref{sec:indexquery} discusses the CP-tree and its related solutions. We report the experimental results in Section~\ref{sec:experiment}, and conclude in Section~\ref{sec:conclusion}.

\section{Related Work}
\label{sec:related}

In the literature, there are two kinds of work related to the retrieval of communities, namely  \emph{community detection} (CD) and \emph{community search} (CS).

\noindent\textbf{Community detection (CD)} aims to obtain all the communities from a given graph. Earlier works~\cite{newman2004finding,newman2006modularity} use link-based analysis to obtain these communities. However, they do not consider the textual information associated with graphs. 
Recent works focus on attributed graphs and use some advanced techniques such as clustering techniques to identify communities. However, these studies often assume that the attribute of the vertex is a set of keywords, and do not consider the hierarchical relationship among them. For Example, Zhou et al.~\cite{zhou2009graph} used keywords to describe vertices and further compute the vertices' pairwise similarities to cluster the graph.
Qi et al.~\cite{qi2013online} studied a problem of dynamically maintaining communities of moving objects using their trajectories.
Ruan et al.~\cite{ruan2013efficient} proposed a method called {\tt CODICIL}. Based on content similarity, {\tt CODICIL} augments the original graphs by creating new edges, and then uses an effective graph sampling to boost the efficiency of clustering.
Another wide-used approach is based on topic models \cite{attr-topic-icml2009,attr-topic-sigmod2012}. Essentially, these methods still analyze the one-dimensional content to obtain the communities.

Another common approach is based on topic models. Link-PLSA-LDA~\cite{liu2009topic} and Topic-Link LDA~\cite{nallapati2008joint} models jointly model vertices’ links and content based on the LDA model. In~\cite{xu2012model}, the communities are clustered based on probabilistic inference. In \cite{sachan2012using}, information such as topics, interaction types and the social connections are considered to explore the communities. CESNA \cite{yang2013community} detects overlapping communities by assuming communities “generate” both the link and content. 
As we introduced before, CD solutions are typically time consuming, and they may not be suitable for online applications that require fast retrieval of communities. It is also interesting to examine how our PCS solutions can be extended to support CD.

\noindent\textbf{Community search (CS)} returns the communities for a given graph vertex in a fast and online manner. Most existing CS solutions~\cite{sozio2010community,online-sigmod2013,k-truss2014,li2015influential, cui2014local} only consider graph topologies, but not the labels associated with the vertices. 
To define the structure cohesiveness of the community, the minimum degree is often used~\cite{cui2014local,li2015influential,sozio2010community}. Sozio et al.~\cite{sozio2010community} proposed the first algorithm {\tt Global} to find the $k$-$\widehat{core}$ containing the query vertex. Cui et al. \cite{cui2014local} proposed {\tt Local}, which uses local expansion techniques to improve {\tt Global}. We will compare these two solutions in our experiments. Other definitions, such as $k$-clique \cite{online-sigmod2013}, $k$-truss \cite{k-truss2014} and edge connectivity \cite{hu2016querying}, have been considered for searching meaningful communities.
Recent CS solutions, such as ACQ~\cite{fang2016effective,fang2017effective} and ATC~\cite{huang2017attribute}, make use of both vertex labels and graph structure to find communities. 

Since CS is ``query-based'', it is much more suitable for fast and online query of the communities on large-scale profiled graphs. However, all above works are not designed for profiled graphs, and they do not consider the hierarchical relationship among vertex labels. Thus in this paper, we propose methods to solve the community search problem on profiled graphs. We have performed detailed experiments on real datasets (Section~\ref{sec:experiment}). As we will show, our algorithms yield better communities than state-of-the-art CS solutions do.

\section{Problem Definition and Basic Solution}
\label{sec:problem}

In this section, we first formally introduce the PCS problem, and then give a basic solution to the PCS problem. Table~\ref{tab:notation} lists all notations used in this paper.

\begin{table}[htp]
\centering \footnotesize \caption {Notations and meanings.}
\label{tab:notation}
  \small
  \begin{tabular}{c|l}
     \hline
          {\bf Notation} & {\bf Meaning}\\
     \hline\hline
          $G(V,E)$       & A profiled graph with vertex set $V$ and edge set $E$\\
     \hline
        $n$     & the number of vertices in $V$\\
    \hline
        $m$  & the the number of edges in $E$\\
     \hline
          $deg_G(v)$     & The degree of vertex $v$ in $G$\\
     \hline
          $T(v)$    & The P-tree of vertex $v$\\
     \hline
          $M(G_q)$ &The maximal common subtree of $G_q$\\
     \hline
           $G[T]$    & \tabincell{c}{the largest connected subgraph of $G$ s.t. \\ $q \in G[T]$ ,
           $\forall v \in G[T]$, $T \subseteq T(v)$}   \\
     \hline
           $G_k[T]$    & \tabincell{c}{the largest connected subgraph of $G[T]$ \\
           s.t. $q \in G_k[T]$, $deg_{G_k}(v) \geq k$}   \\
     \hline
  \end{tabular}
\end{table}

\subsection{The PCS Problem}
\label{PCSproblem}

A profiled community is a subgraph of $G$ that firstly satisfies the structure cohesiveness (i.e., the vertices in this community are connected to each other in some way). Formal definition will be introduced later. A common notion of structure cohesiveness is that the \emph{minimum degree} of all the vertices that in the community has to be at least $k$~\cite{sozio2010community,cui2014local,li2015influential,batagelj2003m,dorogovtsev2006k,seidman1983network}. This is used in the $k$-core and the PC. Let us discuss the $k$-core first.

\begin{definition}[$k$-core~\cite{md1983,batagelj2003m}]
\label{def:kcore}
Given an integer $k$ ($k\geq 0$), the $k$-core of $G$, is the largest subgraph of $G$, such that $\forall v \in k$-core, degree of v is at least $k$.
\end{definition}

Notice that $k$-core may not be connected~\cite{batagelj2003m}. Its connected components, denoted by $k$-$\widehat{core}$, are  the ``communities'' retreieved by $k$-core search algorithms. We use Example~\ref{eg:kcore} to illustrate it. 

\begin{example}
\label{eg:kcore}
In Figure~\ref{fig:pc}, each dashed circle represents a 2-core and also a 2-$\widehat{core}$. Vertices \{$A$, $B$, $D$, $E$\} group a 3-$\widehat{core}$ and vertices \{$A$, $B$, $C$, $D$, $E$\} form a 2-$\widehat{core}$ because $C$ only has a degree of 2, even though other vertices has a higher degree. 
\end{example}

A profiled graph $G(V,E)$ is an undirected graph with vertex set $V$ and edge set $E$. Each vertex $v \in V$ is associated with a \textbf{profiled tree} (P-tree) to describe $v$'s hierarchical attributes.


\begin{definition}[P-tree]
The P-tree of vertex $q$, denoted by $T(q)$= $(V_{T(q)},E_{T(q)})$, is a rooted ordered tree, where $V_{T(q)}$ is the set of attribute labels and $E_{T(q)}$ is the set of edges between labels. A P-tree satisfies following constraints:
(1) There is only one root node $r \in V_{T(q)}$;
(2) $\forall (x,y) \in E_{T(q)}$, it is directed and $y$ is the child attribute label of $x$; and
(3) $\forall y \in V_{T(q)}$ and $y\neq r$, there is one and only one $x\in V_{T(q)}$, s.t. $(x,y) \in E_{T(q)}$.
\end{definition}

In practice, labels in the upper levels of the P-tree are more semantically general than those in lower levels.
All edges in $E_{T(q)}$ preserve the semantic relationships among labels in $V_{T(q)}$.

\begin{definition}[induced rooted subtree]
Given two P-trees $S$=$(V_S,E_S)$ and $T$=$(V_T,E_T)$, $S$ is the induced rooted subtree of $T$, denoted by $S \subseteq T$, if $V_S \subseteq V_T$ and $E_S \subseteq E_T$.
\end{definition}

Essentially, an induced rooted subtree defines an inclusion relationship between two P-trees.
Unless otherwise specified, we use ``subtree'' to mean ``induced rooted subtree''.
 We call the unified P-tree of all vertices' P-trees a \emph{Global P-tree} (GP-tree), which usually corresponds to a taxonomy system in practice.


\begin{definition}[maximal common subtree]
\label{df:MaximalTree}
Given a profiled graph $G$, the maximal common subtree of $G$, denoted by $\mathcal M$($G$), holds the properties:
(1) $\forall v \in G$, $\mathcal M$($G$) $\subseteq T(v)$;
(2) there exists no other common subtree $\mathcal {M'}$($G$) such that $\mathcal {M}$($G$) $\subseteq \mathcal {M'}$($G$).
\end{definition}


The common subtree depicts the common hierarchical part among all P-trees in a subgraph. We use the maximal structure $\mathcal M$($G$) to consider both the high-level and low-level labels and it fully mines the common features of this subgraph. As a result, by using the maximal common subtree, we can maximize vertices' common profiles, including the topology and semantics of users' profiles. 
Next, we formally introduce the PCS problem.

\begin{problem}[\emph{PCS}]
\label{problem:PCS}
Given a profiled graph $G(V,E)$, a positive integer $k$, and a query node $q\in G$, find a set $\mathcal {G}$ of graphs, such that $\forall G_q \in \mathcal {G}$, the following properties hold:

\noindent$\bullet$ \textbf{Connectivity.} $G_q \subseteq G$ is connected and contains $q$;

\noindent$\bullet$ \textbf{Structure cohesiveness.} $\forall v\in G_q$, $deg_{G_q}(v)\geq k$, where $deg_{G_q}(v)$ denotes the degree of $v$ in $G_q$;

\noindent$\bullet$ \textbf{Profile cohesiveness.} There exists no other $G'_q \subseteq G$ satisfying the above two constraints, such that $\mathcal M(G_q) \subseteq \mathcal M(G'_q)$.

\noindent$\bullet$ \textbf{Maximal structure.} There exists no other $G'_q$ satisfying the above properties, such that $G_q \subset G'_q$ and $\mathcal M(G_q)$ = $\mathcal M(G'_q)$;

\end{problem}

Essentially, a profiled community (PC) is a subgraph of $G$, in which vertices are closely related in both structure and semantics.
In Problem~\ref{problem:PCS}, the first two properties and last property ensure the structure cohesiveness, as shown in the literature~\cite{fang2017effective,li2015influential}. The unique property \emph{profile cohesiveness} captures the maximal shared profile among all the vertices of $G_q$. Moreover, since the shared subtree $\mathcal M(G_q)$ shows the common hierarchical attribute, it can well explain the semantic theme of the community.

\vspace{-0.1in}
\subsection{A Basic Solution}
\label{sec:basic}

Since vertices in the PC's share a common subtree of the query vertex $q$, a straightforward method it that we can enumerate all the subtrees of $q$'s P-tree and find the corresponding PC's. However, as illustrated in Lemma~\ref{subtreeNo}, the search space may be exponentially large and computation overhead renders this method impractical. To alleviate this issue, we iteratively perform the following two steps.

\begin{lemma}
\label{subtreeNo}
The maximum number of subtrees of a P-tree with $x$ nodes is $2^{x-1}+1$.
\end{lemma}

\begin{figure}[htp]
\centering
\subfigure[a special case]{
\includegraphics[width=.45\columnwidth]{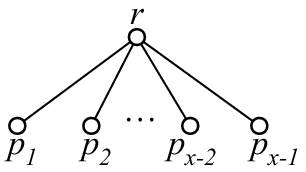}
\label{fig:maximumNumber}
}
\hspace{0.1in}
\subfigure[a general case]{
\includegraphics[width=.45\columnwidth]{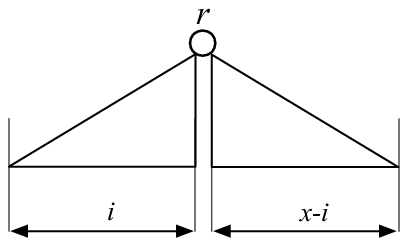}
\label{fig:computeSubtree}
}
\caption{a P-tree with $x$ nodes.}
\end{figure}

\begin{proof}

Let $f(x)=max\{L$$\in\mathbb{N}|$ $L$ is the number of subtrees of a tree with $x$ nodes$\}$.
As shown in Fig.~\ref{fig:maximumNumber}, $p_{i}$ denotes the $i$th child of the P-tree. Then it is not hard to find that there are $(2^{n-1}+1)$ subtrees including the ``empty tree'' (no P-tree node is contained). So $f(x) \geq 2^{x-1}+1$.
In this case, we do need to worry about the ``parent-child'' relationship between P-tree nodes so that $ 2^{x-1}+1$ is also the upper bound of $f(x)$. Then we can infer that $f(x) = 2^{x-1}+1$.

More formally, we can verify the correctness of this formula. As shown in Fig.~\ref{fig:computeSubtree}, the left triangle (including $r$) denotes the subtree with $i$ nodes and the right one represents the subtree with $x-i$ nodes. We present the following equation~\ref{eq:subtreeNoEq1}. Note that the ``empty tree'' should be included and thus $f(0)=1$. Obviously, we can construct different subtrees by combining subtrees in left and right parts. Then we can compute $f(x)$ by using $f(i)$ and $f(x-i)$. Note that the ``empty tree'' in both left and right part should not be included simultaneously. Finally we add 1 to $f(x)$ to represent the ``empty tree''.

\begin{equation}
f(x)=
\begin{cases}
1 & \text{$x=0$}\\
max_{i=0}^x \{f(i) \cdot [f(x-i)-1] \}+1 & \text{$x \geq 1, x\in \mathbb{N}$ }
\end{cases}
\label{eq:subtreeNoEq1}
\end{equation}

Now we can directly verify that $f(x) = 2^{x-1}+1$ satisfy the equation and this complete the proof.

\end{proof}

\noindent\textbf{Step 1: candidate subtree generation.}
To generate the candidate subtrees, the key problem is how to avoid redundancies of the subtree enumeration. In~\cite{asai2004efficient}, Asai et al. introduced a tree pattern enumeration strategy, and it is based on the following two concepts:
(1) \emph{Rightmost leaf} is the last P-tree node according to the depth-first traversal order.
(2) \emph{Rightmost path} is defined as a path from the root node to the rightmost leaf.
Given a tree $T'$, a new subtree $T$ can only be generated by adding a new node $t$ to $T'$ such that the following hold:
(1) $t$'s parent node is on the rightmost path of $T'$;
(2) $t$ is the rightmost leaf of $T$.
As shown in~\cite{asai2004efficient}, this generation strategy guarantees that all the subtrees of the P-tree will be enumerated without repetition. Thus, we follow this strategy to generate the candidate subtrees.

\noindent\textbf{Step 2: community verification.}
After a candidate subtree $T$ has been generated, we verify the existence of the corresponding community.
We use $G_k[T]$ to represent the largest connected subgraph of $G$ containing $q$ where each vertex has at least $k$ neighbors and contains the subtree $T$. We say that, $T$ is \emph{feasible}, if $G_k[T]$ exists.
The verification step is mainly based on the following lemma.

\begin{proposition}
\label{pro:1}
Given a profiled graph $G$, two P-tree $T', T$ and the query vertex $q$, if $T' \subseteq T, G_k[T] \subseteq G_k[T']$.
\end{proposition}

\begin{proof}
As we defined before, $G_k[T]$ denotes the $k$-$\widehat{core}$ containing $q$ where each vertex contains the subtree $T$.
(1) If $G_k[T] = \emptyset$, $G_k[T] \subseteq G_k[T']$ always holds.
(2) If $G_k[T] \neq \emptyset$, we have $\forall v \in G_k[T]$, $T \subseteq T(v)$. Then from $T' \subseteq T$, we can infer $\forall v \in G_k[T]$, $T' \subseteq T(v)$. This means each vertex $v \in G_k[T]$ also contains the P-tree $T'$. Thus if $G_k[T] \neq \emptyset$, $G_k[T] \subseteq G_k[T']$. In summary, Proposition~\ref{pro:1} holds.
\end{proof}

\begin{lemma}[Anti-monotonicity]
\label{lm:anti}
Given a subtree $T$, if $G_k[T] \neq \emptyset$, then $\forall T' \subseteq T$, $G_k[T'] \neq \emptyset$.
\end{lemma}

\begin{proof}
 From Proposition~\ref{pro:1}, we know $\forall T' \subseteq T$, $G_k[T] \subseteq G_k[T']$. Now since $G_k[T] \neq \emptyset$, we have $\forall T' \subseteq T, G_k[T'] \neq \emptyset$.
\end{proof}

By Lemma~\ref{lm:anti}, we can conclude that, if $G_k[T]$ is infeasible, then we can stop generating subtrees from $T$.
The {\tt basic} method begins with generating a subtree from the root node. Then, it iteratively performs the two steps above to retrieve all the feasible $G_k[T]$s, until no larger subtrees can be generated.  Pseudocodes of {\tt basic} are attached in Algorithm~\ref{basic}.

\noindent\textbf{Complexity analysis.} Let $m$ be the number of edges in $G$. In worst case all edges are traversed to compute the $G_k[T]$ and all the subtrees are verified. As a result, {\tt basic} completes in $O$($2^{|T(q)|} \cdot m$) time where $|T(q)|$ denotes the number of nodes of $T(q)$. In practice, the value of $2^{|T(q)|}$ could be exponentially large and this makes {\tt basic} impractical. To alleviate this issue, we propose more efficient index-based solutions in next section.

Algorithm~\ref{basic} presents {\tt basic}.
We first initilize the result set $\mathcal {G}$ and load the $q$'s P-tree $T(q)$ (line 2). Then we need to compute $G_k$, the largest connected subgraph of $G$ containing $q$ where each vertex has at least $k$ degrees (line 3).
Now in the iteration, we generate new subtrees from current subtree $T'$.
For each new subtree $T$, we verify the existence of $G_k[T]$ (lines 4-10).
If $G_k[T]$ exists, we add $T$ in $\Phi$ (lines 11-12); otherwise if no subtree can be generated from $T'$ or all subtrees generated from $T'$ are infeasible, we add $G_k[T']$ in $\mathcal{G}$ if $T'$ is maximal (line 13).
Finally, all PC's are returned (line 14).

\begin{algorithm}[htp]
\caption{{\tt basic} query algorithm}
\label{basic}
\footnotesize{
\algrenewcommand{\algorithmiccomment}[1]{\hskip3em$//$ #1}
\begin{algorithmic}[1]
\Function{query($G,q,k$)}{}
  \State $\mathcal {G} \gets \emptyset$, load $T(q)$ from $G$;
  \State compute $G_k$ from $G$;
  \If{$G_k \neq \emptyset$ }
      \State$\Psi \gets$ \Call{generateSubtree($\emptyset,T(q)$)}{};
      \While{$\Psi \neq \emptyset$}
           \State $T'\gets \Psi.pop()$; flag $\gets true$;
           \State $\Phi \gets$ \Call{generateSubtree($T',T(q)$)}{};
       \For{each $T \in \Phi$}
            \State compute $G_k[T]$ from $G_k$;
             \If{$G_k[T] \neq \emptyset$}
                 \State flag $\gets false$; $\Psi.push(T)$;
            \EndIf
      \EndFor
          \If{flag = $true$ and $T'$ is maximal} 
              \State $\mathcal {G}=\mathcal {G} \cup G_k[T']$;
          \EndIf
     \EndWhile
  \EndIf
  \State \Return $\mathcal {G}$;
\EndFunction

\end{algorithmic}
}
\end{algorithm}

\section{Index-based Solutions}
\label{sec:indexquery}
We first introduce some preliminaries and the proposed \textbf{CP-tree} index, and then discuss the index-based query algorithms.

\vspace{-0.1in}
\subsection{$k$-core and CL-tree}

\textbf{\boldmath{$k$-core}.}
In line with existing CS~\cite{li2015influential,fang2016effective}, we use $k$-core to satisfy the constraints of minimum degree and maximal structure of a PC.
Given an integer $k$ ($k\geq 0$), the $k$-core of $G$, denoted by $G_{k}$, is the largest subgraph of $G$, such that $\forall v \in G_k$, $deg_{G_k}(v) \geq k$.
Since $G_k$ may be disconnected, we use $k$-$\widehat{core}$s to denote one of its connected components.
An important property of $k$-core is the ``nested" property: given two integer $i$ and $j$, $j$-$\widehat{core} \subseteq$ $i$-$\widehat{core}$  if $i < j$.
In Fig.~\ref{fig:kcores}, the $0$-core represents the whole graph, and 3-core is nested in 2-core.
Computing all the $k$-cores of a graph $G$, known as core decomposition, can be completed by an $O$($m$) algorithm~\cite{batagelj2003m}, where $m$ is the number of edges in $G$.

\textbf{CL-tree.}
Since $k$-cores are nested, all the $k$-cores of a graph can be organized into a tree structure, called \textbf{CL-tree}~\cite{fang2016effective}. In this paper, we adopt it, but skip the labels on the tree.
The CL-tree of the graph in Fig.~\ref{fig:kcores} is shown in Fig.~\ref{fig:Ctree}. Clearly, vertices in each CL-tree node and other vertices in all its descendant nodes represent a $k$-$\widehat{core}$.
For example, vertex $C$ and other vertices $\{A, B, D, E\}$ in its child node compose a 2-$\widehat{core}$. Since each vertex appears only once, the space cost of CL-tree is $O(n)$ where $n$ is the number of vertices in $G$.
In addition, we maintain a map \emph{vertexNodeMap}, where the key is the vertex and the value is the node of the corresponding CL-tree node, and it allows us to locate the $k$-$\widehat{core}$ containing any query vertex efficiently.

\vspace{-0.1in}
\begin{figure}[htp]
\centering\mbox{
 \subfigure[$k$-cores]{
            \includegraphics[width=.4\columnwidth]{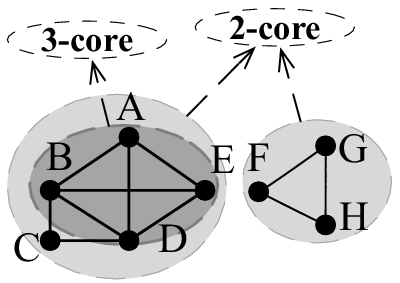}
            \label{fig:kcores}
  }
  \hspace{0.1in}
 \subfigure[CL-tree]{
            \includegraphics[width=.4\columnwidth]{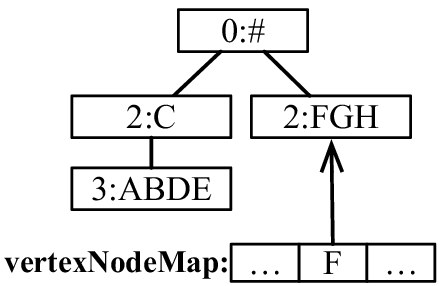}
            \label{fig:Ctree}
  }
}
\caption{$k$-cores, CL-tree.}
\label{fig:KTree}
\end{figure}

\vspace{-0.1in}
\subsection{CP-tree Index}

\noindent\textbf{Index Overview.} We build the \underline{C}ore \underline{P}rofiled tree (CP-tree) index by considering both the P-tree structure and $k$-cores.
We depict an example CP-tree in Fig.~\ref{fig:CP-tree} using the profiled graph in Fig.~\ref{fig:pg}.

\begin{figure}[htp]
\centering\mbox{

\includegraphics[width=.6\columnwidth]{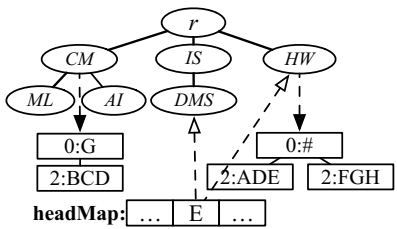}

}
\caption{CP-tree index.}
\label{fig:CP-tree}
\end{figure}

Each CP-tree node corresponds to a label and stores the $k$-cores sharing this label.
To summarize, each node $p$ consists of following four elements:

(1) \emph{label}: the attribute label;

(2) \emph{parentNode}: the parent node of $p$;

(3) \emph{childList}: a list of child CP-tree nodes of $p$; and

(4) \emph{vertexNodeMap}: a map that stores the CL-tree.

In addition, we maintain a map \emph{headMap}, where the key is a vertex $v$, and the value is a list of CP-tree nodes, each of which corresponds to a leaf node of $v$'s P-tree. Main advantages of CP-tree are listed below.

\noindent$\bullet$ \textbf{Restore P-trees.} By utilizing the \emph{headMap}, each vertex's P-tree can be restored by traversing the leaf nodes up to the root node.

\noindent$\bullet$ \textbf{Locating \boldmath{$k$-$\widehat{core}$}.} Given an integer $k$, a query vertex $q$ and a CP-tree node $t$, using \emph{vertexNodeMap}, we design a function $get(k,q,t)$ to get the $k$-$\widehat{core}$ containing $q$ where each vertex contains the label $t.label$ in constant time cost. 

\noindent$\bullet$ \textbf{Query efficiency.} As discussed above, the label information of each vertex's P-tree can be efficiently accessed using the \emph{headMap}.


\noindent\textbf{Index Construction.} We incrementally create CP-tree nodes and then link them up to build the CP-tree index. Pseudocodes of CP-tree index construction are presented in Algorithm~\ref{CPtree}. For each vertex $v$, we read $T(v)$ and create new CP-tree nodes (lines 2-5). For each CP-tree node $t$, we add $v$ in $t$ for later CL-tree construction (lines 6, 9). If P-tree node $x$ is a leaf node, we update $headMap$ (line 7). Then we link up all CP-tree nodes according to the GP-tree structure. Note that if GP-tree is unknown, we can simultaneously unify it whiling reading P-trees in the previous step (line 10). Finally, $\mathcal{I}$ is returned (line 11).

\begin{algorithm}[htp]
\caption{CP-tree index construction}
\label{CPtree}
\footnotesize{
\algrenewcommand{\algorithmiccomment}[1]{\hskip3em$//$ #1}
\begin{algorithmic}[1]
\Function{buildIndex($G(V,E)$)}{}
	\For{each $v \in V$}
		\For{each $x \in T(v)$}
			\State $t \gets$ a CP-tree node in $\mathcal{I}$ such that $t.label$ = $x.label$;
			\If{$t $ = $null$} create a CP-tree node $t$ and add it in $\mathcal{I}$; 
				\EndIf
			\State add $v$ in $t$;
			\If{$x$ is the leaf node of $T(v)$} $headMap.put(v,t)$; \EndIf
		\EndFor
	\EndFor
    \For{each $t \in \mathcal{I}$}
 		\State Build CL-tree for the subgraph of $t$;
 		\State link to its parent and child nodes;
 	\EndFor

    \State \Return $\mathcal{I}$;
\EndFunction
\end{algorithmic}
}
\end{algorithm}

\textbf{Complexity analysis.}
Obviously, lines 2-7 take the linear time. The time complexity of building a CL-tree is $O$($m \cdot \alpha(n)$)~\cite{fang2017effective,fang2016effective} where $m$ is the number of edges in $G$ and $\alpha(n)$, the inverse Ackermann function, is less than 5 for large value of $n$. Thus the time complexity of building CP-tree is $O$($|P|\cdot m \cdot \alpha(n)$), and it is linear to the size of $G$. The space cost of CP-tree is $O$($|P|\cdot n$) where $|P|$ denotes the number of labels in $G$. The space cost of the \emph{headMap} is $O$($\hat{l}\cdot n$) where $\hat{l}$ denotes the average number of leaf nodes in each vertex's P-tree and $\hat{l} < |P|$. Therefore, the total space complexity is $O$($|P|\cdot n$) which is linear to the size of $G$.

\vspace{-0.1in}
\subsection{Index-based Query Algorithms}

Now we present our index-based query solutions.
The first one follows the framework of {\tt basic}, and it incrementally generates and verifies the subtrees of P-tree (from smaller subtrees to larger ones). Thus we call it {\tt incre}.
The advanced methods borrows some ideas from MARGIN~\cite{thomas2010margin}, the algorithm of mining {\it maximal frequent subgraphs}.
As we will explain later, advanced methods can find all PC's by examining a small fraction of subtrees, resulting in high efficiency. 
In addition, their time complexities are $O$($2^{|T(q)|} \cdot m$), because in the worst case all the subtrees are verified. However, as we will show in Section~\ref{sec:efficiency}, in practice they are much more efficient than such worse-case time complexities.

\vspace{-0.1in}
\subsubsection{The Method {\tt incre}}



We begin with an interesting lemma, which greatly accelerates the verification step. 

\begin{lemma}
\label{lm:incre}
Given a CP-tree index $\mathcal{I}$, a subtree $T'$ and a new subtree $T$ which is generated from $T'$ by adding a new P-tree node. We have $G_k[T] \subseteq G_k[T'] \cap \mathcal{I}.get(k,q,T$$\setminus$$T')$, where $T$$\setminus$$T'$ denotes the new added node.
\end{lemma}

\begin{proof}
$T=T'\cup t$, so we have $T' \subseteq T$. Based on Proposition~\ref{pro:1}, we know $G_k[T] \subseteq G_k[T']$. Similarly, $t \subseteq T$, then we have that $G_k[T] \subseteq \mathcal{I}.get(k,q,T$$\setminus$$T')$ where $\mathcal{I}.get(k,q,T$$\setminus$$T')$ is the $k$-$\widehat{core}$ containing the query vertex $q$ and P-tree node $T$$\setminus$$T'$. Hence $G_k[T] \subseteq G_k[T'] \cap \mathcal{I}.get(k,q,T$$\setminus$$T')$.
\end{proof}

As {\tt incre} searches the communities in the subgraph which are found in former iteration, the query efficiency is improved. 
We present {\tt incre} in Algorithm~\ref{incre}.

\begin{algorithm}[htp]
\caption{{\tt incre} query algorithm}
\label{incre}
\footnotesize{
\algrenewcommand{\algorithmiccomment}[1]{\hskip3em$//$ #1}
\begin{algorithmic}[1]
\Function{query($\mathcal{I},q,k$)}{}
	\State restore $T(q)$ using $\mathcal{I}.headMap$;
	\State $\mathcal {G} \gets \emptyset$,$\Psi \gets$ \Call{generateSubtree($\emptyset,T(q)$)}{};
	\While{$\Psi \neq \emptyset$}
		\State $T'\gets \Psi.pop()$; flag $\gets true$;
		\State $\Phi \gets$ \Call{generateSubtree($T',T(q)$)}{};
			\For{each $T \in \Phi$}
				\State compute $G_k[T]$ from $G_k[T'] \cap \mathcal{I}.get(k,q,T$$\setminus$$T')$;
				\If{$G_k[T] \neq \emptyset$}
					\State flag $\gets false$; $\Psi.push(T)$;
				\EndIf
			\EndFor
		\If{flag = $true$ and $T'$ is maximal} 
			\State $\mathcal {G}=\mathcal {G} \cup G_k[T']$;
		\EndIf
	\EndWhile
	\State \Return $\mathcal {G}$;
\EndFunction

\end{algorithmic}
}
\end{algorithm}

We first use \emph{headMap} to locate the leaf nodes of $T(q)$ and then restore $T(q)$ (line 2).
We initialize $\Psi$ by using $T(q)$ (line 3).
In the iteration, for current subtree $T'$, we generate new subtrees.
For each new subtree $T$, we verify the existence of $G_k[T]$ using the index (lines 4-8).
If $G_k[T]$ exists, we add $T$ in $\Phi$ (lines 9-10); otherwise if no subtree can be generated from $T'$ or all subtrees generated from $T'$ are infeasible, we add $G_k[T']$ in $\mathcal{G}$ if $T'$ is maximal (line 11).
Finally, all PC's are returned (line 12). 

\vspace{-0.1in}
\subsubsection{The {\tt Advanced} Methods}{

The method {\tt incre} follows the Apriori-based method, which explores all possible subtrees by traversing the search space from smaller subtrees to larger ones; 
while, as demonstrated in the Section~\ref{setup}, the maximal feasible subtrees often lie in the middle of the search space, which implies that most of the exploration may be avoided. Based on this observation, we adapt MARGIN~\cite{thomas2010margin} to tackle PCS.

\noindent\textbf{MARGIN:} It does not perform a bottom-up (or top-down) traversal of the search space; instead, it narrows the search space by examining only subgraphs that lie on the border of frequent and infrequent subgraphs.
It firstly finds an initial pair of graphs ($CR$, $R$) where $R$ is frequent and $CR$ is not.
In addition, $CR$ is the {\it child subgraph} of $R$ (i.e., $CR$ is the subgraph of $R$ and they differ by exactly one edge). Similarly, $R$ is the {\it parent subgraph} of $CR$.
($CR$, $R$) is called a {\tt cut} and from this {\tt cut}, MARGIN expands and finds all other {\tt cuts} by adding or deleting an edge to obtain new adjacent subgraphs. MARGIN defines this function as {\tt expandCut} and Thomas et al.~\cite{thomas2010margin} has proved that {\tt expandCut} is able to find all maximal frequent subgraphs.

Inspired by MARGIN, we design the following functions.

\textbf{1. Function~{\tt expandPtree}.} This function is adapted from {\tt expandCut}~\cite{thomas2010margin} and the main modifications are as follows.

\noindent$\bullet$ We dynamically obtain child subgraphs and parent sugraphs, which are called \emph{child subtrees} and \emph{parent subtrees} in our case, using the \emph{parentNodes} and \emph{childLists} of CP-tree nodes, instead of pre-computing all subtrees in the search space as MARGIN does.

\noindent$\bullet$ We define a pair of P-trees ($IF$,$F$) as a cut, where $IF$ is the child subtree
of $F$ and $F$ is feasible while $IF$ is not;

\noindent$\bullet$ We dynamically verify whether a feasible subtree is maximal.

\noindent$\bullet$ We develop a function {\tt verifyPTree}  to verify the feasibility.

\begin{algorithm}[htp]
\caption{{\tt expandPtree}}
\label{ag:expandPtree}
\footnotesize
\algrenewcommand{\algorithmiccomment}[1]{\hskip3em$//$ #1}
\begin{algorithmic}[1]
\Function{expandPtree($IF,F,\mathcal {G}$)}{}
	\If{$IF = \emptyset$ and $F \neq \emptyset$} update $\mathcal {G}$;
	\Else
		\State$Q \gets \emptyset$; $Q.push((IF,F))$;
		\While{$Q \neq \emptyset$}
			\State $(IF,F) \gets Q.pop()$;
			\For{each parent $Y_i$ of $IF$}
				\If{$Y_i$ is feasible}
					\State update $\mathcal {G}$ if $Y_i$ is maximal;
					\For{each child $K$ of $Y_i$}
						\If{$K$ is infeasible} $Q.push((K,Y_i))$;\EndIf
						\If{$K$ is feasible}
							\State find common child $C$ of $K$ and $IF$;
							\State $Q.push((C,K))$;
						\EndIf
					\EndFor
				\Else
					\For{each parent $K$ of $Y_i$}
						\If{$K$ is feasible} $Q.push((Y_i,K))$; \EndIf
					\EndFor
				\EndIf
			\EndFor
		\EndWhile
	\EndIf
	\State \Return $\mathcal {G}$;
\EndFunction
\end{algorithmic}
\end{algorithm}

We now illustrate {\tt expandPtree} in Algorithm~\ref{ag:expandPtree}.
As we will introduce later, if $IF = \emptyset$ and $F \neq \emptyset$ we can directly update $\mathcal {G}$ because the $F$ is already the maximal common subtree (line 2).
Otherwise, we first use ($IF,F$) to initialize the queue $Q$ (line 4).
Then, for each pair, we iteratively verify its adjacent pairs (lines 5-17).
If the parent subtree $Y_i$ of $IF$ is feasible, $G_k[Y_i]$ here may not be the final result.
This is because subtrees are not regularly enumerated, and thus $Y_i$ may be temporarily maximal, so we need to repeatedly verify it.  If there exist other feasible subtrees verified in previous steps that are the subtree of $Y_i$, we need to replace their corresponding subgraphs with $G_k[Y_i]$ (line 9).
Finally, we return $\mathcal {G}$ (line 18).

\begin{lemma}
\label{lm:expandPtree}
Given a P-tree pair ($IF,F$), {\tt expandPtree} can find all feasible subtrees for a PCS query.
\end{lemma}

The proof of Lemma~\ref{lm:expandPtree} is based on following preliminaries.

\begin{figure}[htp]
\label{fig:margin}
\centering\mbox{
		\subfigure[lattice]{
            \includegraphics[width=.4\columnwidth]{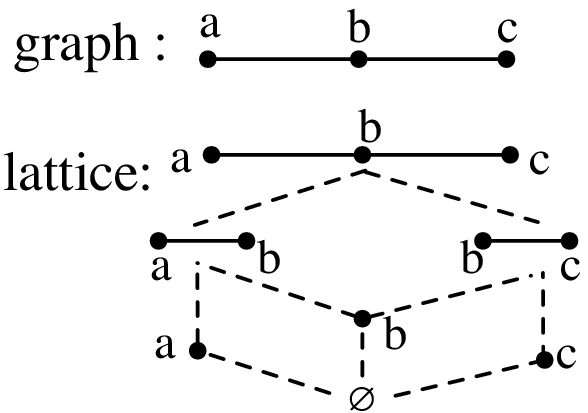}
            \label{fig:lattice}
        }
        \hspace{0.2in}
        \subfigure[Upper-$\Diamond$-Property]{
            \includegraphics[width=.4\columnwidth]{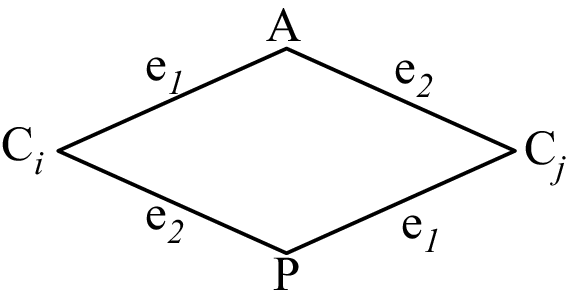}
			\label{fg:upperDiamond}
        }
}
\caption{the lattice and Upper-$\Diamond$-Property~\cite{thomas2010margin}.}
\end{figure}

\textbf{Lattice} is essentially a pre-processed data structure where all possible subgraphs of a given graph are enumerated. Taking the graph in Fig.~\ref{fig:lattice} as an example, its subgraphs in each level have the same size (i.e., numbers of edges). The bottom level (level 0) corresponds to the empty graph and the level $i$ lists all size-$i$ subgraphs. In lattice, each subgraph is linked to its \emph{parent} graphs (i.e., subgraph of this graph and they differ exactly by one edge) and \emph{child}s (i.e., super-graph of this graph and they differ exactly by one edge). We can observe that the P-tree can directly replace the graph to construct the lattice.

\begin{property}[Upper-$\Diamond$-Property~\cite{thomas2010margin}]
\label{property:upper}
Any two child subgraphs $C_i, C_j$ of a graph $P$ will have a common child subgraph $A$.
\end{property}
In Property~\ref{property:upper}, $C_i, C_j, P$ and $A$ are four subgraphs. $C_i, C_j$ are two \emph{child} subgraphs of $P$ (i.e., subgraphs of $P$ and they respectively differ with $P$ by one egde $e_1, e_2$). Then there must exist one subgraph $A$ such that $A$ is the child subgraph of $C_i$ and $C_j$. Property~\ref{property:upper} is very intuitive in graphs. Based on Proposition\ref{pro:2}, we prove that the Upper-$\Diamond$-Property can be simply adapted to fit in P-tree models.

\begin{proposition}
\label{pro:2}
P-trees satisfy the Upper-$\Diamond$-Property.
\end{proposition}

\begin{proof}
In P-trees, $e_1$ and $e_2$ can be two P-tree nodes such that subtrees $C_i=P \cup e_1$ and $C_j=P \cup e_2$. There must exist a P-tree $A = P \cup e_1 \cup e_2 = (P \cup e_1) \cup e_2 = (P \cup e_2) \cup e_1$. Thus $A = C_i \cup e_2 = C_j \cup e_1$ which means $A$ is the common child subtree of $C_i$ and $C_j$.
\end{proof}

Now we formally give the proof of Lemma~\ref{lm:expandPtree}.

\begin{proof}
Method {\tt expandPtree} is mainly adpted from MARGIN. As mentioned in MARGIN, the correctness holds when the adapted problem satisfies the following constraints~\cite{thomas2010margin}:

\begin{description}
\item[(1)]The search space is a subset of the \emph{lattice}.
\item[(2)]The \emph{Upper}-$\Diamond$-\emph{property} holds.
\item[(3)]The anti-monotone property is satisfied.
\item[(4)]A candidate set can be defined which is a ``boundary'' set
such that every in the set satisfies a given user-constraint and there exists an immediate child in the lattice that does not satisfy the constraint because of the anti-monotone property. For every in the set, there exists an immediate parent that does not satisfy the constraint for the monotone property.
\item[(5)]Solution sets can be generated from the candidate sets.
\end{description}

For PCS problem, the ``element'' in constraint (1) is the P-tree and obviously constraint (1) is satisfied. Proposition~\ref{pro:2} has proved that constraint (2) is satisfied. The anti-monotonicity property has been proved in Lemma~\ref{lm:anti} and thus constraint (3) is also satisfied.
In MARGIN, the ``user-constraint'' of the constraint (4) is that, given a threshold, whether a graph is frequent or not.
Here for constraint (4), the ``user-constraint'' is that whether a P-tree is feasible. For instance, a P-tree $T'$ is feasible which means $G_k[T']$ exists. If $T$, which is the child of $T'$, is not feasible (i.e., $G_k[T']$ does not exist). Then $T'$ can be defined in this ``boundary'' set and its immediate child $T$ does not satisfy this ``user-constraint'' for the anti-monotone property. Hence constraint (4) holds.
Once a is added in the candidate set, we need to verify whether this is maximal. It means the solution set is the subset of this candidate set. Thus constraint (5) is satisfied.
In conclusion, the correctness of Lemma~\ref{lm:expandPtree} holds.
\end{proof}

\textbf{2. Function~{\tt verifyPtree}.}
Given a subtree $T$, $T_{child}$ and $T_{parent}$ denote a child and the parent subtree of $T$.
Let $l$ denote the number of $T_{parent}$'s leaf nodes and $t_{n_i}$ represent the $i$th leaf node of $T_{parent}$.
Derived from Lemma~\ref{lm:incre}, we have

\noindent$\bullet$ \ \ {$G_k[T_{child}] \ \subseteq \ G_k[T] \cap \mathcal{I}.get(k,q,T_{child}$$\setminus$$T)$}.

\noindent$\bullet$ \ \ $G_k[T_{parent}] \ \subseteq \ \bigcap_{i=1}^{l} \mathcal{I}.get(k,q,t_{n_i})$.

Since all P-trees are subtrees of the GP-tree, if a P-tree has the attribute $t$, then $t$'s parent attribute $t'$ is also included. Thus, $\mathcal{I}.get(k,q,t) \subseteq \mathcal{I}.get(k,q,t')$. For a special subtree $T_i$ (a path from leaf node $t_{n_i}$ to root node $r$), we can finally get $G_k[T_i]=\mathcal{I}.get(k,q,t_{n_i})$. Note that $T_{parent}$ can be seen as several paths and thus we get $G_k[T_{parent}]\subseteq \bigcap_{i=1}^{l} \mathcal{I}.get(k,q,t_{n_i})$.

Based on CP-tree, {\tt verifyPtree} can efficiently verify subtrees. Next we discuss three methods to find the initial cut.

\textbf{3. Function~{\tt find-I}.} We can adapt {\tt incre} to find the initial cut. As shown in Algorithm~\ref{ag:find-I}, we incrementally enumerate subtrees and verify the existence of the corresponding communities. Once we find a subtree which is feasible while its child subtree is not, then we can regard them as an initial cut (lines 2-15).

\begin{algorithm}[htp]
\caption{Find the initial cut: {\tt find-I}}
\label{ag:find-I}
\footnotesize
\algrenewcommand{\algorithmiccomment}[1]{\hskip3em$//$ #1}
\begin{algorithmic}[1]
\Function{find-I($\mathcal{I},S,q,k$)}{}
	\State restore $T(q)$ using $\mathcal{I}.headMap$;
	\State $IF \gets \emptyset$; $F = T(q)$;
	\State $\Psi \gets$ \Call{generateSubtree($\emptyset,T(q)$)}{};
	\While{$\Psi \neq \emptyset$}
		\State $T'\gets \Psi.pop()$; flag $\gets true$;
		\State $\Phi \gets$ \Call{generateSubtree($T',T(q)$)}{};
			\For{each $T \in \Phi$}
				\State compute $G_k[T]$ from $G_k[T'] \cap \mathcal{I}.get(k,q,T$$\setminus$$T')$;
				\If{$G_k[T] \neq \emptyset$}
					\State flag $\gets false$; $\Psi.push(T)$;
				\EndIf
			\EndFor
		\If{flag = $true$ and $T'$ is maximal} 
			\State $F = T'$; $IF = T$;
			\State break;
		\EndIf
	\EndWhile
\State\Return$(IF,F)$;	
\EndFunction
\end{algorithmic}
\end{algorithm}

\textbf{4. Function~{\tt find-D}.} We can decrementally generate subtrees from larger subtrees to smaller ones. We represent {\tt find-D} pseudocodes in Algorithm~\ref{ag:find-D}. Firstly, if $G_k[T(q)]$ exists, we can directly return it as a qualified community (lines 2-4). In each step, for an infeasible subtree $T$, we remove one of $T$'s leaf nodes and verify the feasibility of the new subtrees (lines 6-11). Once there is a new feasible subtree, we treat $T$ and this new subtree as the initial cut (lines 12-17). 

\begin{algorithm}[htp]
\caption{Find the initial cut: {\tt find-D}}
\label{ag:find-D}
\footnotesize
\algrenewcommand{\algorithmiccomment}[1]{\hskip3em$//$ #1}
\begin{algorithmic}[1]
\Function{find-D($\mathcal{I},S,q,k$)}{}
	\State $IF\gets \emptyset$; $F\gets \emptyset$;
	\State restore $T(q)$ using $\mathcal{I}.headMap$;
	\If{$G_k[T(q)] \neq \emptyset$} $F = T(q)$;
	\Else
		\State $\Psi.push(T)$; 
		\While{$\Psi \neq \emptyset$}
			\State $T\gets \Psi.pop()$; $IF = T$;
			\State $\Theta \gets$ all leaf nodes of $T$; 
			\For{each $t \in \Theta$}
				\State compute $G_k[T\setminus t']$ from $G$;
				\If{$G_k[T\setminus t'] \neq \emptyset$} 
					\State $F = T\setminus t'$;
					\State Break;
				\Else
					\State $\Psi.push(T\setminus t')$;
				\EndIf
			\EndFor
		\EndWhile
		\EndIf
\State\Return$(IF,F)$;	
\EndFunction
\end{algorithmic}
\end{algorithm}

\textbf{4. Function~{\tt find-P}.} We can find the initial cut by directly verifying subtrees instead of the node one by one. Intuitively, P-tree can be divided into several paths (from leaf nodes to the root). According to Lemma~\ref{lm:anti}, these paths can be further verified by checking the corresponding leaf nodes. We call it find initial cut by path ({\tt find-P}).

\begin{algorithm}[htp]
\caption{Find the initial cut: {\tt find-P}}
\label{ag:find-P}
\footnotesize
\algrenewcommand{\algorithmiccomment}[1]{\hskip3em$//$ #1}
\begin{algorithmic}[1]
\Function{find-P($\mathcal{I},S,q,k$)}{}
	\State $IF\gets \emptyset$; $F \gets$ find a leaf node $t\in S$ s.t. $\mathcal{I}.get(k,q,t) \neq \emptyset$;
	\If{$F\neq \emptyset$}
		\For{each $t \in S$}
			\State computing $G_k[F\cup t]$ from $G_k[F] \cap \mathcal{I}.get(k,q,t)$;
			\If{$G_k[F\cup t] \neq \emptyset$}
				$F=F\cup t$;
			\Else
				\State $path \gets $ trace a path from $t$ to $r$ in $\mathcal{I}$;
				\State find $t', t'_{parent}$ on $path$ s.t. $G_k[t']$=$\emptyset$, $G_k[t'_{parent}] \neq \emptyset$;
				\State $IF=F\cup t'_{parent}$; $F=F\cup t'$;
				\State Break;
			\EndIf

		\EndFor

	\Else
		\For{each $t \in S$} $S.replace(t,t.parent$); \EndFor
		\State\Call{find-P($\mathcal{I},S,q,k$)}{};
	\EndIf
\State complete subtrees $IF,F$;	
\State\Return$(IF,F)$;	
\EndFunction
\end{algorithmic}
\end{algorithm}

We present the pseudocodes of {\tt find-P} in Algorithm~\ref{ag:find-P}.
$S$ denotes a P-tree node set. Initially, it consists of all leaf nodes of $T(q)$.
If there does not exist a feasible node in $S$, we trace up to verify their parent nodes (lines 13-14).
Next, we iteratively check the nodes in $S$.
If we find a node $t$ and $G_k[F\cup t]$ exists, we update $F$ (lines 5-6).
Let $t'_{parent}$ denote the parent node of $t'$.
If we find a node $t$ that $G_k[F\cup t]$ does not exist, we trace up to find the ``boundary'' where $G_k[t'_{parent}]$ exists while $G_k[t']$ does not and thus we find an initial pair (lines 8-11).
Note that at now stage, $IF$, $F$ may not be complete subtrees. Thus for the nodes in $IF$ and $F$, we need to include all their ancestor nodes and then return $(IF,F)$ as a cut (lines 15-16).

Algorithm~\ref{ag:advanced} gives the overall {\tt advanced} methods.
Notice that, there are three functions, i.e., {\tt find-I}, {\tt find-D}, and {\tt find-P}, of finding the initial cut, so we have three variants of {\tt advanced}, denoted by {\tt adv-I}, {\tt adv-D} and {\tt adv-P} respectively.

\begin{algorithm}[htp]
\caption{Advanced method}
\label{ag:advanced}
\footnotesize
\algrenewcommand{\algorithmiccomment}[1]{\hskip3em$//$ #1}
\begin{algorithmic}[1]
\Function{query($\mathcal{I},q,k$)}{}
		\State $\mathcal {G} \gets \emptyset$;
		\State $(IF,F) \gets$ \Call{find($\mathcal{I},S,q,k$)}{};
		\State \Call{expandPtree($IF,F,\mathcal {G}$)}{};
		\State \Return $\mathcal {G}$;
\EndFunction
\end{algorithmic}
\end{algorithm}

\section{Experiments}
\label{sec:experiment}
\subsection{Setup}
\label{setup}

We consider two real datasets (ACMDL and PubMed) and two synthetic datasets (Flickr and DBLP).
{\it ACMDL}~\footnote{\url{https://dl.acm.org/}} and {\it PubMed}~\footnote{\url{https://www.nlm.nih.gov}} are the co-authorship networks of researchers in computer science and biomedical areas respectively.
Each vertex of them represents an author, and an edge is a co-authorship between two authors.
For each author, her papers have been categorized by a hierarchical subject classification system (\emph{ACM CCS} or \emph{Medical Subject Headings (MeSH)}~\footnote{\url{https://meshb.nlm.nih.gov/}}), so we build the P-tree by unifying the categorization information of all her papers.
For {\it Flickr}~\footnote{\url{https://www.flickr.com/}}~\cite{thomee2015new}, each vertex represents a user and each edge denotes a ``follow'' relationship between two users. For {\it DBLP}~\footnote{\url{http://dblp.uni-trier.de/xml/}}, a vertex is an author and an edge represents a co-authorship relationship.
For each user, we use a hash function and map the associated textual content to subjects of CCS to synthesize a P-tree. By doing this, the same textual contents could be mapped for constructing the same nodes in P-trees.
Table~\ref{tab:dataset} shows the statistics of the datasets, including the numbers of vertices and edges, vertices' average degree {$\widehat d$}, the average number of labels in P-trees $\widehat P$, and the average number of labels in the GP-tree.

To evaluate PCS queries, in line with~\cite{fang2016effective}, we set the default value of $k$ to 6. For each dataset, we randomly select 100 query vertices from the 6-core.
We implement all the algorithms in Java, and run experiments on a machine having an eight-core Intel 3.40GHz processor, and 16GB of memory, with Ubuntu installed.

\vspace{-0.1in}
\begin{table}[htp]
  \centering
  \footnotesize \caption {Datasets used in our experiments.}\label{tab:dataset}
  \begin{tabular}{|c|r|r|c|c|c|}
     \hline
          {\bf Dataset}  & \multicolumn{1}{c|}{\textbf{Vertices}}
                         & \multicolumn{1}{c|}{\textbf{Edges}}
                         & \textbf{\emph{{$\widehat d$}}}
                         & \textbf{\emph{{$\widehat P$}}}
                         & \textbf{$\vert$ GP-tree  $\vert$}\\
     \hline\hline
          ACMDL        &  107,656      &  717,958     & 13.34   &  11.54 & 1,908 \\
     \hline
          Flickr           &581,099       & 4,972,274    &  17.11  &  26.63 &1,908\\
     \hline
          PubMed        &  716,459    &  4,742,606  &  13.22  &  27.10 &10,132\\
     \hline
          DBLP        &  977,288   &  6,864,546    &  14.04  &  37.98 &1,908 \\
     \hline

  \end{tabular}
\end{table}

we consider all the four datasets and check the locations of maximal feasible subtrees of 100 communities in search space for each dataset. In our experiments, because the search space may be very large, according to the depth, we average them into 5 levels. Notice that, in this case, level 3 represents the middle location of the search space. The experimental results are attached below. For example, there are 43\% maximal feasible subtrees lying on the middle of the search space in PubMed. This demonstrates the above view and explains the motivation for the advanced methods. 

\begin{table}[htp]
  \centering
  \footnotesize \caption {Locations of maximal feasible subtrees.}\label{tab:locations}
  \begin{tabular}{|c|c|c|c|c|}
     \hline
            & \textbf{ACMDL} & \textbf{Flickr} & \textbf{PubMed} & \textbf{DBLP}\\
     \hline\hline
          Level 1       & 3\%   &  8\%  & 11\% &5\% \\
     \hline
          Level 2       &  15\%  &  23\% & 5\% &13\% \\
     \hline
          Level 3       &  18\%  &  32\% & 43\% &37\% \\
     \hline
          Level 4       &  26\%  &  25\% & 24\% &31\% \\
    \hline
          Level 5       &  38\%  &  12\% & 17\% &14\% \\
  \hline
  \end{tabular}
\end{table}

\vspace{-0.1in}
\subsection{PCS Effectiveness}

As mentioned before, the existing CS methods mainly focus on non-attributed graphs. A recent work ACQ~\cite{fang2017effective,fang2016effective} investigates CS on attributed graphs. In ACQ, each vertex in the attributed graph is associated with a set of keywords. 
Communities retrieved by ACQ should satisfy the structure cohesiveness ($k$-core constraint) and ``keyword cohesiveness''~\cite{fang2016effective,fang2017effective}, i.e., the number of common keywords shared by all vertices in communities should be maximum. 
We compare PCS with ACQ. To run ACQ queries, we set each vertex's attribute as a set of keywords, which are the keywords in its P-tree.
In the following, we first present a case study, and then show the quality and diversity of communities.

\noindent$\bullet$ {\bf A Case Study}: We perform a case study on the ACMDL dataset and consider a renowned researcher: Jim Gray. We set $k$ = 4 here.
We present Jim's two PC's, i.e., PC1 and PC2, with different research areas in Fig.~\ref{fig:caseStudy1} and Fig.~\ref{fig:caseStudy2}. 
Notice that ACQ only finds one community PC1 shown in Fig.~\ref{fig:jim11}.
This is because, ACQ maximizes the number of shared keywords, so PC2 shown in Fig.~\ref{fig:jim21}, which has five shared keywords, cannot be returned.
In addition, as shown in Fig.~\ref{fig:jim12}, all shared keywords of PC1 are organized in a tree with few branches, which implies that the semantics of keywords are highly overlapped with each other.
In contrast, the shared subtree of PC2 shown in Fig.~\ref{fig:jim22} has multiple branches, so the semantics of keywords are very different and diversified.
Hence, PCS are more effective than ACQ for extracting communities from profiled graphs.


\begin{figure}[htp]
\centering
 \mbox{
    \subfigure[PC1]{
        \includegraphics[width=.38\columnwidth]{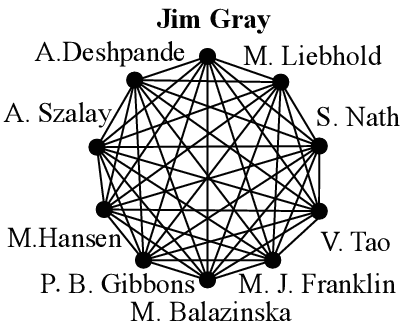}
        \label{fig:jim11}
    }
    \subfigure[The maximal common subtree of PC1]{
        \includegraphics[width=.55\columnwidth]{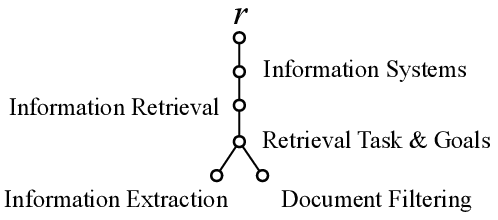}
          \label{fig:jim12}
     }
    }
    \caption{One PC of Jim Gray.}\label{fig:caseStudy1}
\end{figure}

\begin{figure}[htp]
\centering 
 \mbox{
    \subfigure[PC2]{
        \includegraphics[width=.38\columnwidth]{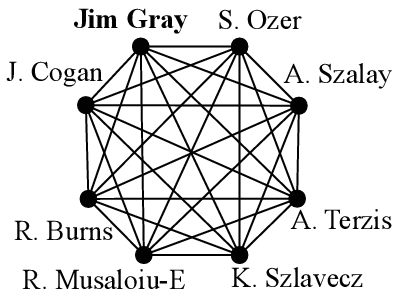}
        \label{fig:jim21}
    }
    \subfigure[The maximal common subtree of PC2]{
        \includegraphics[width=.55\columnwidth]{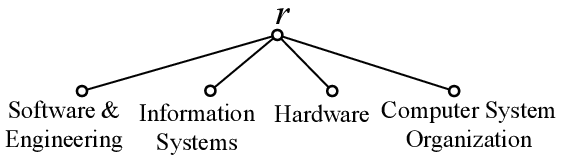}
           \label{fig:jim22}
     }
    }
    \caption{Another PC of Jim Gray.}\label{fig:caseStudy2}
\end{figure}



 \noindent$\bullet$ {\bf Community Pairwise Similarity (CPS)}: We compare PCS with three classic CS methods using ``minimum degree'' definition: ACQ~\cite{fang2016effective}, {\tt Global}~\cite{sozio2010community} and {\tt Local}~\cite{cui2014local}. We use Tree Edit Distance (TED) to compute the similarity between the P-trees of any pair of vertices in community $G_l$. Let $T_i$ be the P-tree of the $i$-th vertex in $G_l$. The CPS is then the average similarity over all pairs of $G_l$'s vertices, and all communities of $\mathcal{G}$:

\begin{equation}
\label{eq:CPS}
CPS(\mathcal{G}) = 1- \sum\limits_{l = 1}^{\cal |\mathcal{G}|} \Bigg [ \frac{1}{|G_l|^2} \sum\limits_{j=1}^{|G_l|} \sum\limits_{i=1}^{|G_l|} \frac{TED(T_i,T_j)}{|T_i\cup T_j|}   \Bigg]
\end{equation}

The CPS($\mathcal{G}$) value has a range of 0 and 1. The higher the value is, the more cohesive the community is. As shown in Fig~\ref{fig:CPS}, PCs$^*$ denotes the communities that only PCS can search. P-ACs represents those returned by both of PCS and ACQ. P-ACs have the most P-tree nodes (i.e., keywords in ACQ definition) in common, and the fewest vertices. Thus they have the highest CPS values. Note that PCs$^*$ have a close CPS value with P-ACs which implies that these unique PC's are also of highly quaility.

\noindent$\bullet$ {\bf Level-diversity ratio (LDR)}: To further measure the quality of PC's, we define a metric, called \emph{level-diversity ratio} (LDR), to measure the diversity of attributes level by level in the shared subtrees. $F$ denotes the method that we use here to compare with PCS. 
Given a query vertex $q$, we use $\mathcal{T}({F,q,j})$ to represent the maximal common P-trees of $j$-th community returned by the method $F$. $\cal L$ is the number of levels in P-tree $T(q)$. $L_i(T)$ is the number of unique labels in the $i$-th level of P-tree $T$.  ${\cal H}$ and ${\cal J}$ denote the numbers of communities returned by the method $F$ and PCS respectively. 
A lower LDR value implies that the method $F$ is less diverse than PCS.

\begin{equation}
\label{eq:LDR}
LDR(q,F) = \frac{1}{\cal L} \sum\limits_{i = 1}^{\cal L} \frac{\sum\limits_{h = 1}^{\cal H} L_i\Big[ \mathcal{T}({F,q,h)} \Big]}
{\sum\limits_{j = 1}^{\cal J} L_i\Big[\mathcal{T}({PCS,q,j)}\Big ] }
\end{equation}

Intuitively, LDR reflects the proportion of unique labels in each level. The experimental results are depicted in Fig.~\ref{fig:LDR}, which shows that communities returned by ACQ can only cover 40\% to 60\% labels of PC's in each level. This implies that PC's found by PCS have higher diversity than those of ACQ, because PCS focuses on maximizing the common structure of P-trees, rather than the number of common keywords. As a result, all communities with the semantically maximal properties can be found, and the communities are of high diversity.

\noindent$\bullet$ {\bf Community numbers}: Fig.~\ref{fig:communityNum} reports the average number of communities that per query request returns in these methods.  
From the results, we can see that PCS finds more communities than others.
This is because only PCS focuses on profiled graphs and hierarchical information in P-trees to retrieve communities. Comapred with other methods, PCS is able to extract communities with more semantic focuses.

\noindent$\bullet$ {\bf Community P-tree Frequency (CPF)}: CPF is inspired by the document frequency measure. Let $fre_{i,j}$ represent the number of vertices in $G_i$ whose P-tree contains $T(q)$'s $j$-th P-tree node. We use CPF to compute the occurence frequency over all nodes in $T(q)$ and all communities in $\mathcal{G}$:

\begin{equation}
\label{eq:CPS}
CPF(q) = \frac{1}{|G_l|\cdot |T(q)|} \sum\limits_{i=1}^{|G_l|} \sum\limits_{j=1}^{|T(q)|} \frac{fre_{i,j}}{|G_i|}
\end{equation}

Note that CPF($q$) ranges from 0 to 1 and a higher value implies a better cohesiveness. As shown in Fig~\ref{fig:CPS}, compared with the communties retrieved by both of PCS and ACQ, those unique PCs also have a highly degree of cohesiveness. 

\begin{figure}[htp]
    \centering
    \mbox{
        \subfigure[CPS]{
            \includegraphics[width=.48\columnwidth]{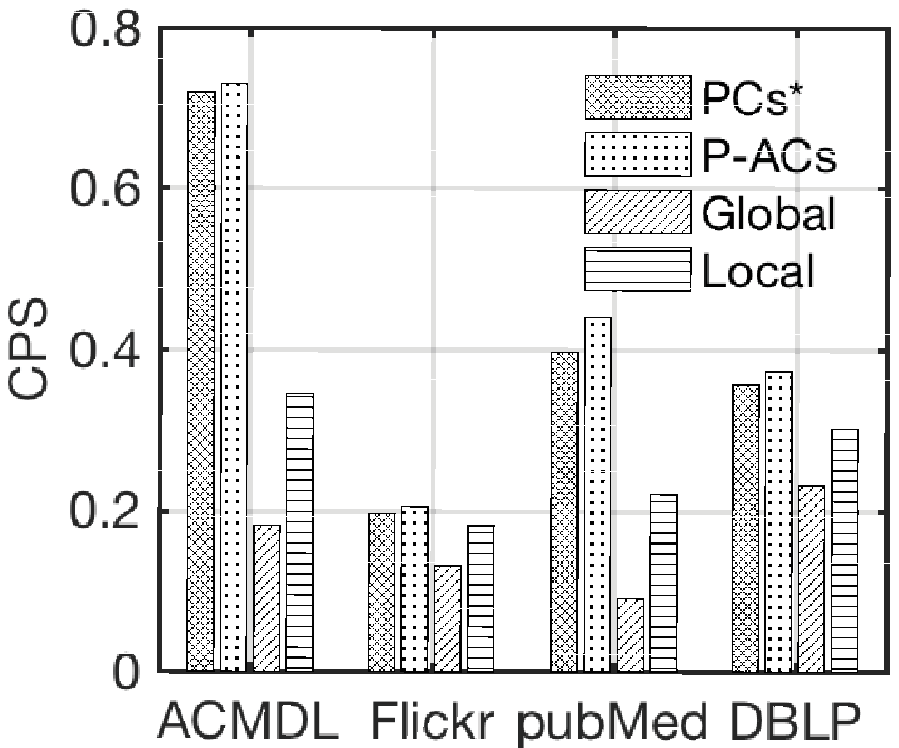}
            \label{fig:CPS}
        }
       \hspace{-0.1in}
         \subfigure[LDR]{
            \includegraphics[width=.48\columnwidth]{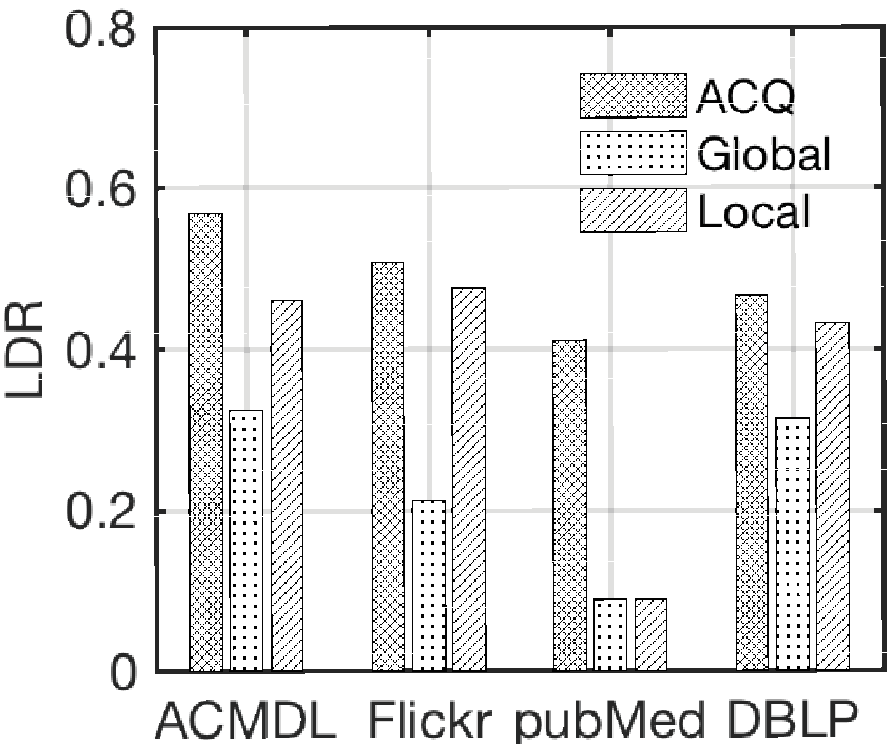}
            \label{fig:LDR}
        }
    }
    \caption{Comparing PCS with CS methods.}\label{fig:effectiveness2}
\end{figure}

\begin{figure}[htp]
    \centering
    \mbox{
        \subfigure[Community number]{
            \includegraphics[width=.48\columnwidth]{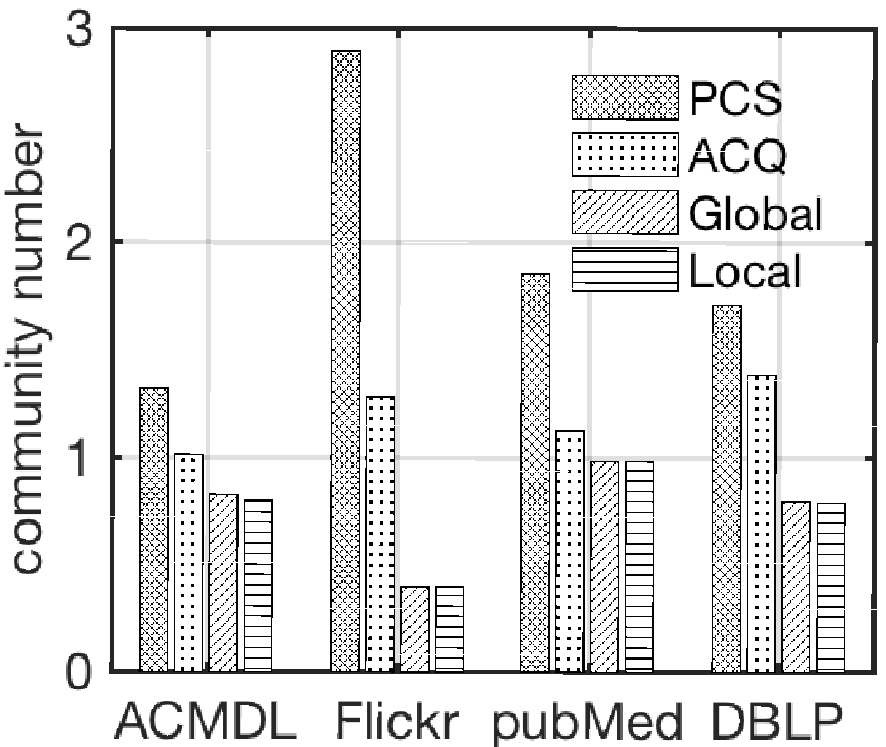}
            \label{fig:communityNum}
        }
        \hspace{-0.15in}
        \subfigure[CPF]{
            \includegraphics[width=.48\columnwidth]{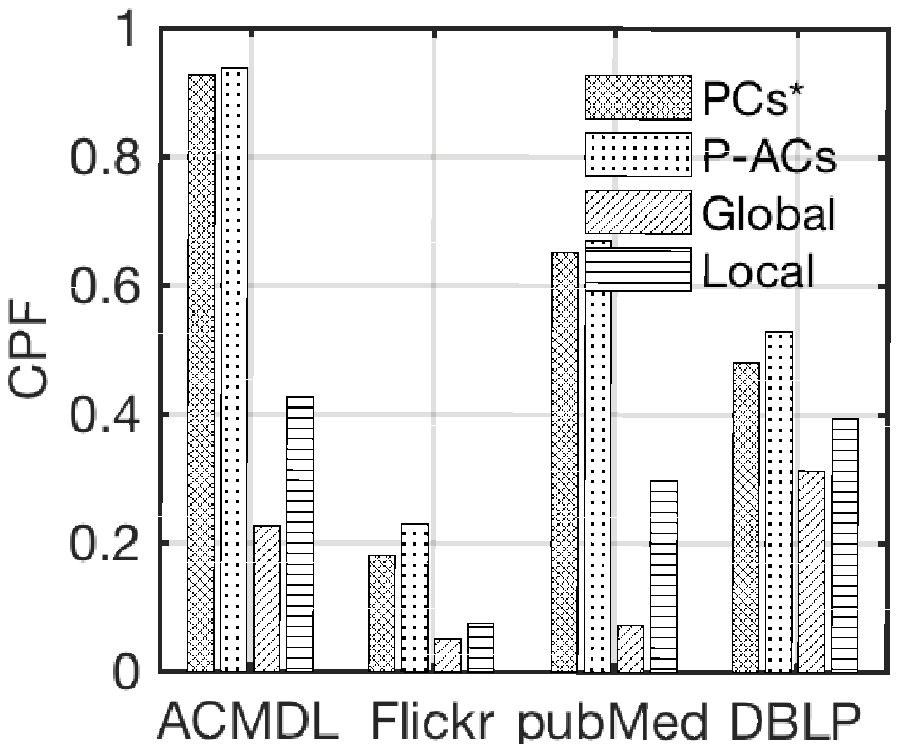}
            \label{fig:CPF}
        }
        
    }
    \caption{Comparing PCS with CS methods.}\label{fig:effectiveness2}
\end{figure}

\noindent$\bullet$ {\bf F1-score}: Here we use Facebook ego-networks~\footnote{\url{http://snap.stanford.edu/}} to evaluate the accuracy. We use FBX to denote the X-th network and each ego-network has several overlapping ground-truth communities, called friendship circles~\cite{leskovec2012learning}. 
See Table~\ref{tab:FBdataset}, each vertex has real profiles, such as political, education, etc. Similar to Flickr, we build each P-tree by using a hash function to map the real profiles to CCS subjects. We random query 100 vertices in these ground-truth communities and compute the F1-scores~\footnote{\url{https://en.wikipedia.org/wiki/F1 score}} over different methods. The F1-scores of all methods over three networks are shown in Fig.~\ref{fig:F1}. The experimental results show that,  compared with other methods, PCS can stably extract communities with high accuracy over three real networks.   

\begin{table}[htp]
  \centering
  \footnotesize \caption {Facebook datasets.}\label{tab:FBdataset}
  \begin{tabular}{|c|r|r|c|c|}
     \hline
          {\bf Dataset}  & \multicolumn{1}{c|}{\textbf{Vertices}}
                         & \multicolumn{1}{c|}{\textbf{Edges}}
                         & \textbf{\emph{{$\widehat d$}}}
                         & \textbf{\emph{{$\widehat P$}}}\\
                         
     \hline\hline
          FB1        &  1,233     &  11,972     & 19.41   &  34.54 \\
     \hline
          FB2       & 1,447       & 17,533   &  24.23  &  29.12 \\
     \hline
          FB3       &  982    &  10,112  &  20.59  &  31.10 \\
     \hline

  \end{tabular}
\end{table}

\begin{figure}[htp]
\centering
\includegraphics[width = 0.5\columnwidth]{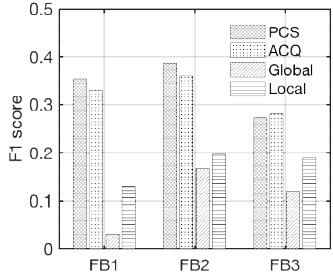}
\caption{F1-scores over three networks.}\label{fig:F1}
\end{figure}
\vspace{-0.1in}
\subsection{Comparison with Other Definition Metrics}
In this section, we compare several potential metrics to define the PCS problems. Generally, a good community should be a group of users, which are cohesive in both structures and profiles. To measure structure cohesiveness, we use the minimum degree metric, which is in line with existing works \cite{cui2014local,li2015influential,sozio2010community,fang2016effective,huang2017attribute}. To measure the profile cohesiveness, we have tried a list of possible metrics, including: 

\noindent(a) common nodes of P-trees;

\noindent(b) common path of P-trees (from the P-tree leaf to the root);

\noindent(c) common subtree of P-tree structures;

\noindent(d) similarity of vertex P-trees.

We compare these four metrics over two real datasets (ACMDL and pubMed). As shown in Fig~\ref{fig:expDef}, compared with other metrics, Metric (c) can achieve highest scores over four indices. We now discuss the reason for such differences. 

\begin{figure}[htp]
\centering
\begin{tabular}{c c}
\begin{minipage}{4.2cm}
  \includegraphics[width=4.2cm]{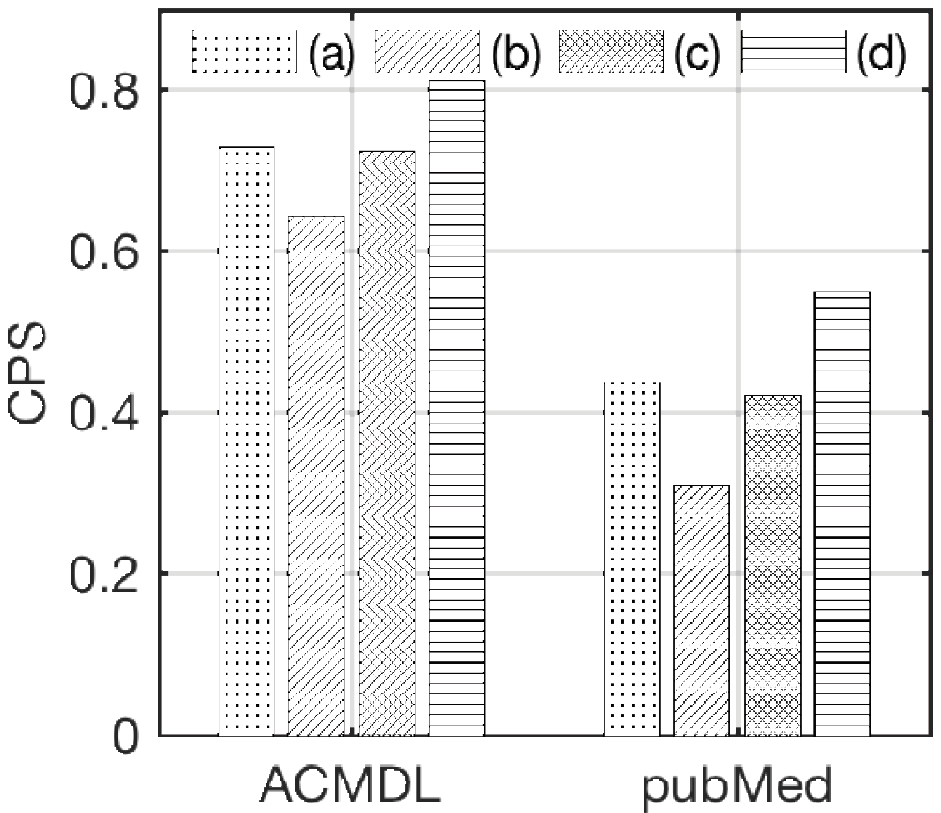}
  \end{minipage}
  &
  \begin{minipage}{4.2cm}
  \includegraphics[width=4.2cm]{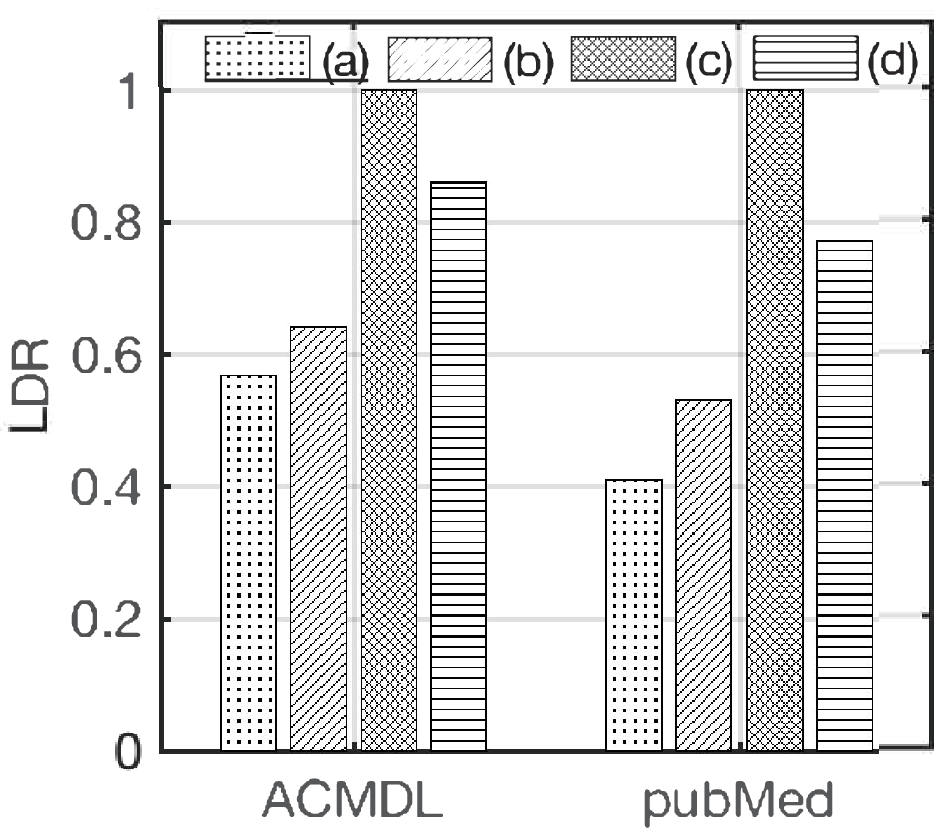}
  \end{minipage}
  \\
  \hspace{+4ex}\scriptsize (a) CPS
  &
  \hspace{+4ex}\scriptsize (b) LDR
  \\
 \begin{minipage}{4.2cm}
  \includegraphics[width=4.2cm]{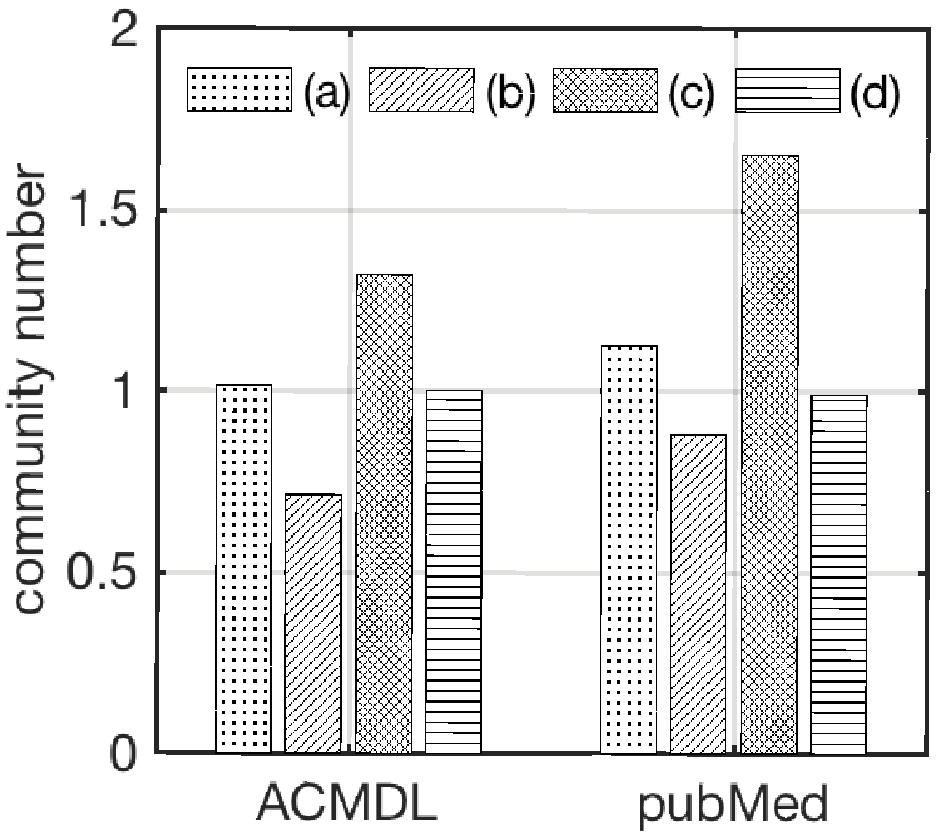}
  \end{minipage}
  &
  \begin{minipage}{4.2cm}
  \includegraphics[width=4.2cm]{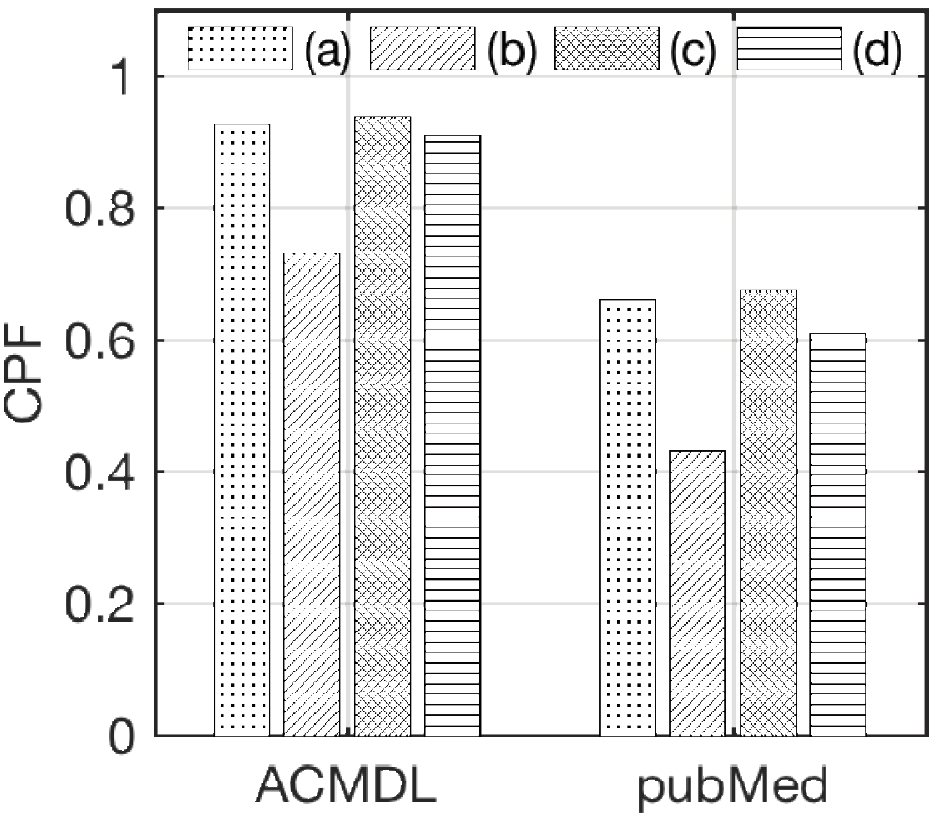}
  \end{minipage}
  \\
  \hspace{+4ex}\scriptsize (c) community number 
  &
  \hspace{+4ex}\scriptsize (b) CPF
  \\
\end{tabular}
\caption{\scriptsize  Evaluation on ACMDL, PubMed datasets.}
\label{fig:expDef}
\end{figure}

In a recently work ACQ~\cite{fang2017effective}, the authors define the vertex attribute as a set of keywords and use the number of “shared keywords” to constrain the communities. Thus, in our PCS problem, it is natural to use the number of common P-tree nodes to measure the profile cohesiveness, and it is natural to require the number of common nodes to be the largest. However, as we have analyzed before, this will ignore the interrelated relations among the nodes and violate the basic motivation for the PCS problem. Thus Metric (a) is not suitable for PCS definition. Metric (b) is defined by common paths (i.e., a common path from P-tree root to a leaf node) shared by all the nodes in the returned community. Intuitively, we can require the number of common paths to be maximum. This metric will still have some inadequacies, as it amounts to maximize the number of common leaf nodes, which will miss out meaningful communities with fewer common leaves. As a result, based on the discussions above, we think metric (b) is also not suitable for PCS problem definition. Metric (c) focuses on the common subtree of all P-trees. Clearly, a sub-tree consists a set of nodes and their hierarchical relationships. Compared with the metrics above, the common subtree of P-tree structure is more suitable for measuring the profile cohesiveness of a community, as it can adequately present the commonalities of vertex P-trees. Inspired by another recent community search work~\cite{huang2017attribute}, we tried to use the similarity of P-trees to define the problem. It means, given a threshold, to find all vertices with a budgeted similarity score. However, it is still not suitable for the PCS problem. This is because, normally, if two P-trees are to be compared by some similarity methods, the diversity of these P-trees will be nevertheless regarded as the dissimilarity. Thus, based on above discussion and experimental results in Fig~\ref{fig:expDef}, we adopt Metric (c) in our PCS problem definition. 

\vspace{-0.1in}
\subsection{Results of Efficiency Evaluation}
 \label{sec:efficiency}
In this section, we show the efficiency results of index construction and PCS queries. 

\noindent\textbf{1. Index construction.} Fig.~\ref{fig:exp-cs-tkde1}(a)-\ref{fig:exp-cs-tkde1}(b) show the scalability of the CP-tree index construction method. To evaluate the scalability of index construction method w.r.t the dataset size, 
for each dataset, we randomly select 20\%, 40\%, 60\% and 80\% of its vertices to obtain four sub-datasets respectively. As shown in Fig.~\ref{fig:exp-cs-tkde1}(a), we observe that, the time cost of the index construction is linear to the size of profiled graphs, which confirms our analysis before. Furthermore, to evaluate the scalability of index construction method over different P-tree sizes of vertices and over different fractions of the GP-tree size, we obtain four sub-datasets in a similar way. As shown in Fig.~\ref{fig:exp-cs-tkde1}(b) and Fig.~\ref{fig:exp-cs-tkde1}(c), we demonstrate that the time cost of the index construction is linear to the size of P-trees and GP-trees. 

\begin{figure*}[htp]
\centering
\begin{tabular}{c c c}
\begin{minipage}{4.175cm}
  \includegraphics[width=4.175cm]{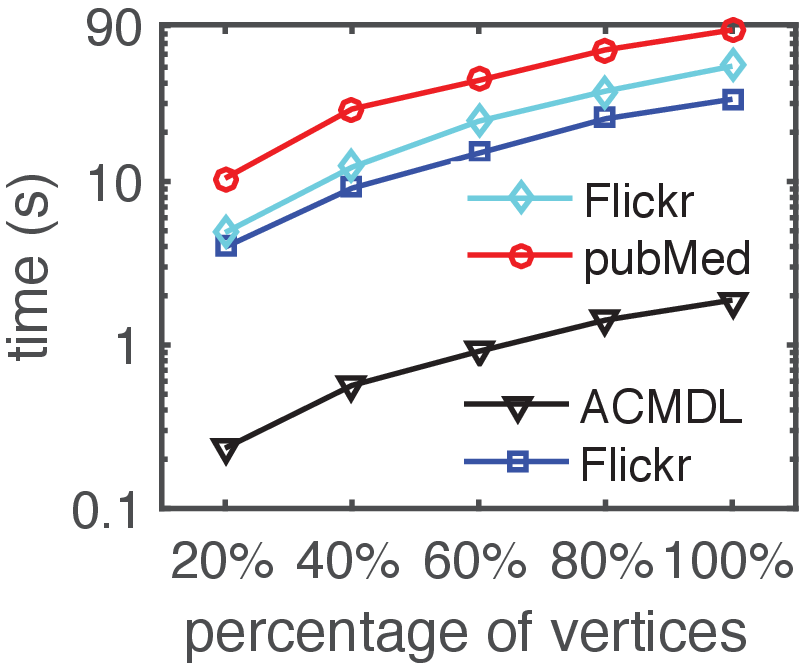}
  \end{minipage}
  &
  \begin{minipage}{4.175cm}
  \includegraphics[width=4.175cm]{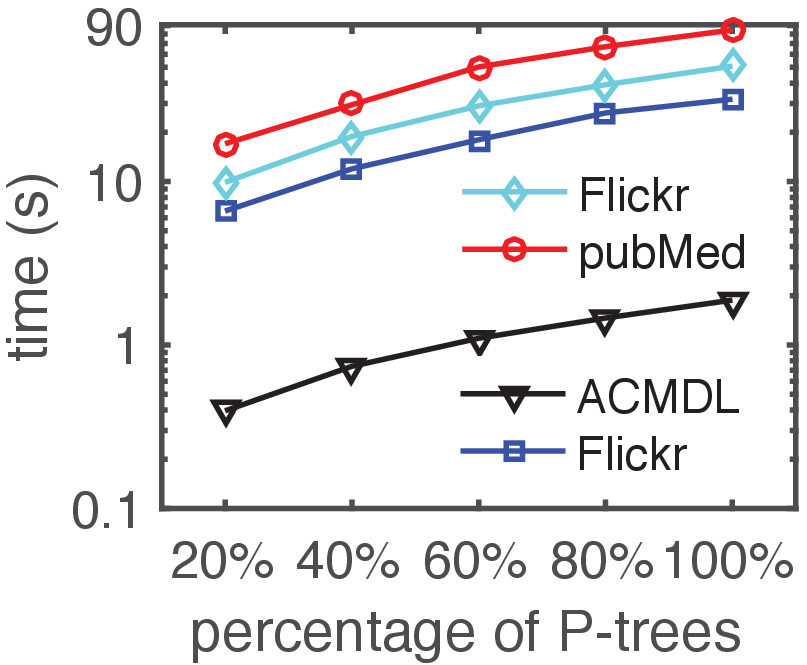}
  \end{minipage}
  &
  \begin{minipage}{4.175cm}
  \includegraphics[width=4.175cm]{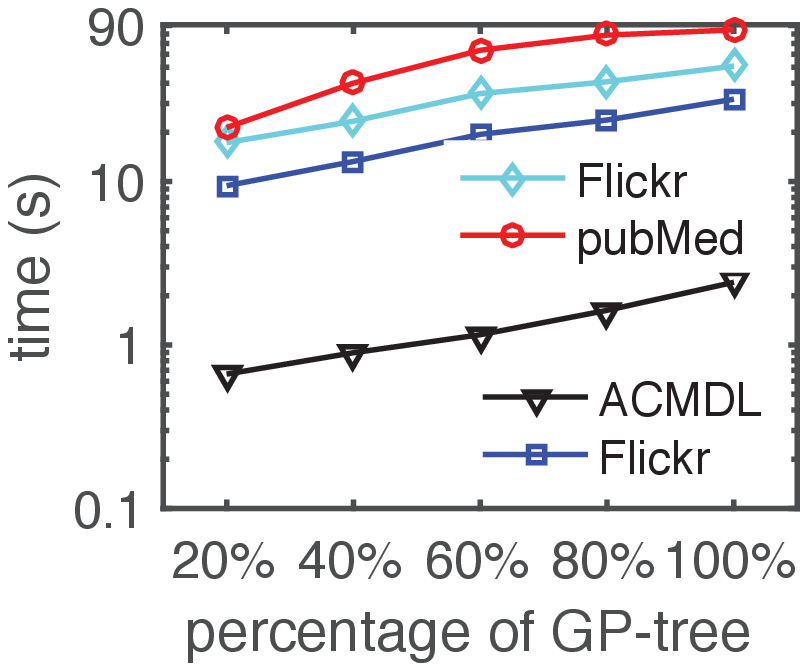}
  \end{minipage}
  \\
  \hspace{+4ex} (a) Vertex
  &
  \hspace{+4ex} (b) P-tree
  &
  \hspace{+4ex} (c) GP-tree
\end{tabular}
\caption{ Efficiency and scalability of index construction on ACMDL, Flickr, PubMed and DBLP datasets.}
\label{fig:exp-cs-tkde1}
\end{figure*}

\noindent\textbf{2. Query efficiency.} We vary the value of $k$ and show the query efficiency of different algorithms in Fig.~\ref{fig:exp-cs-tkde2}(a)-\ref{fig:exp-cs-tkde2}(d). The method {\tt incre} is 100 times faster than the {\tt basic} method, but slower than the method {\tt adv-I}. Further, {\tt adv-D} and {\tt adv-P} are 10 times faster than {\tt incre}. 
The reason is that, compared with {\tt incre}, the advanced methods narrow the search space by verifying a smaller fraction of subtrees. Also, the efficiency gap in finding an initial cut results in the sightly different performance of the advanced methods. Thus, the index-based methods run fast and {\tt adv-P} stably scales the best.
Note that three advanced methods perform similarly on Flickr. This is because the initial cut results are in the middle of the search space. Thus they have similar performance even though they search from different directions.

\noindent\textbf{3. Scalability w.r.t. vertex.} Fig.~\ref{fig:exp-cs-tkde2}(e)-\ref{fig:exp-cs-tkde2}(h) report the scalability over different fraction of vertices. For each dataset, we randomly select 20\%, 40\%, 60\% and 80\% of its vertices to respectively obtain four sub-datasets. Note that in this experiment, vertices' P-trees are fully considered. As shown in the experimental results, the algorithms run slower as more vertices involved. From the above analysis and experimental results, {\tt basic} has been proved to be quite inefficient and we do not involve it afterwards. As show in Fig.~\ref{fig:exp-cs-tkde2}(e)-\ref{fig:exp-cs-tkde2}(h), {\tt incre} and {\tt adv-I} sacle simlilarly and {\tt adv-I} is slightly better. This is because after finding a feasible answer, {adv-I} will quickly search all answers instead of exploring the whole search space. Not surprisingly, {\tt adv-D} and {\tt adv-P} scale the best.

\noindent\textbf{4. Scalability w.r.t. P-tree.} Fig.~\ref{fig:exp-cs-tkde2}(i)-\ref{fig:exp-cs-tkde2}(l) examine scalability over different fraction of P-trees for each vertex. For the P-tree of each vertex, we randomly select 20\%, 40\%, 60\% and 80\% of its P-tree nodes to generate the corresponding subtree. Here all vertices are considered. As shown in Fig.~\ref{fig:exp-cs-tkde2}(i)-\ref{fig:exp-cs-tkde2}(l), {\tt adv-I} performs better than {\tt incre}. This is because {\tt adv-I} avoids exploring the whole search space after finding an initial solution, which accelarates the query process.   
Also, {\tt adv-D} and {\tt adv-P} stably perform the best and {\tt adv-P} performs slightly better than {\tt adv-D}. The reason is that, as we introduced before, {\tt adv-P} finds initial P-tree cuts by directly verifying a batch of P-tree nodes rather than verifying nodes one by one. Thus {\tt adv-P} is always faster than {\tt adv-D}.

\noindent\textbf{5. Scalability w.r.t. GP-tree.} We test the importance of GP-tree size in Fig.~\ref{fig:exp-cs-tkde2}(m)-\ref{fig:exp-cs-tkde2}(p). For the GP-tree of each dataset, we randomly select 20\%, 40\%, 60\% and 80\% of its P-tree nodes to generate new GP-trees. Here we consider all the vertex and P-trees. As GP-tree varies, again, {\tt adv-I} is faster than {\tt incre}. Moreover, {\tt adv-D} and {\tt adv-P} achieve the best performance.

\noindent\textbf{6. Effect of {\tt find} functions}. By varying $k$ from 4 to 8, we compare three {\tt find} functions of the advanced methods in Fig.~\ref{fig:exp-cs-tkde2}(q)-\ref{fig:exp-cs-tkde2}(t). {\tt find-P} and {\tt find-D} are about 10-100 times faster than {\tt find-I}. This is because, as we expained before, {\tt find-P} can find initicial cut by directly verifying the leaf nodes of subtrees instead of enumerating nodes from the root node one by one. Note that in Fig.~\ref{fig:exp-cs-tkde2}(r), {\tt adv-D} and {\tt adv-P} have similar performance. The reason is that the efficiency of finding initial cuts depends on the distribution of initial cuts in search space. Actually, {\tt adv-D} search initial cuts incrementally and {\tt adv-P} explores the search space by decrementally verifying subtrees. Thus {\tt adv-D} and {\tt adv-P} may have similar performance.

\begin{figure*}[htp]
\centering
\begin{tabular}{c c c c}
  \begin{minipage}{4.175cm}
  \includegraphics[width=4.175cm]{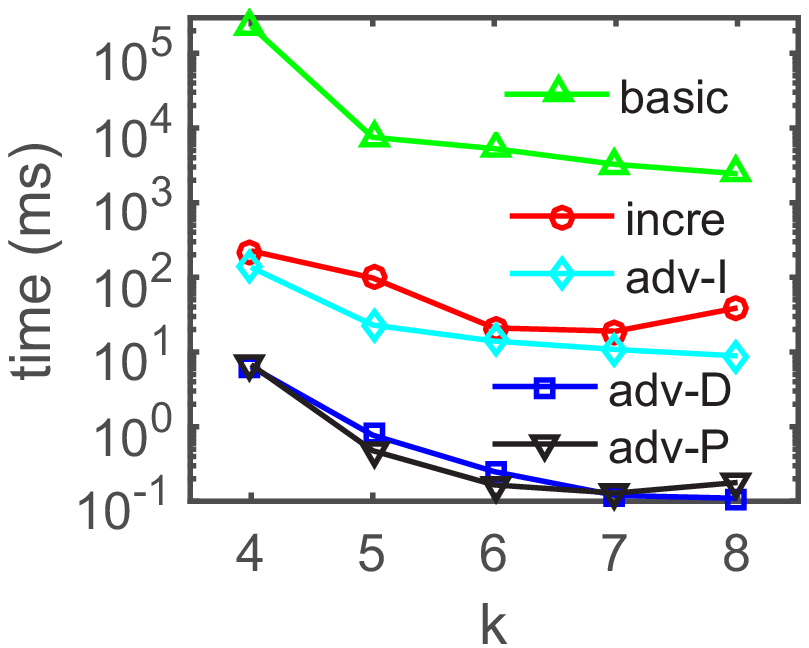}
  \end{minipage}
  &
  \begin{minipage}{4.175cm}
  \includegraphics[width=4.175cm]{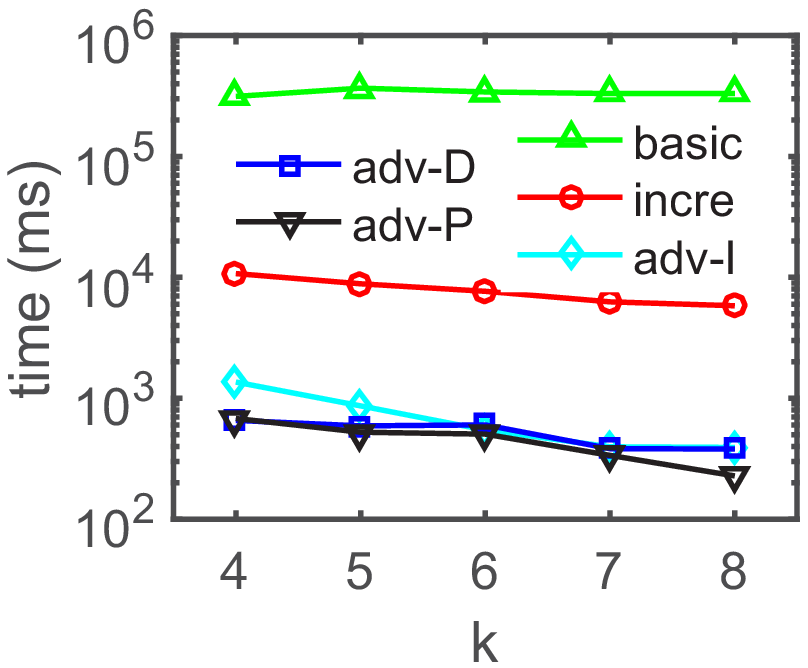}
  \end{minipage}
  &
  \begin{minipage}{4.175cm}
  \includegraphics[width=4.175cm]{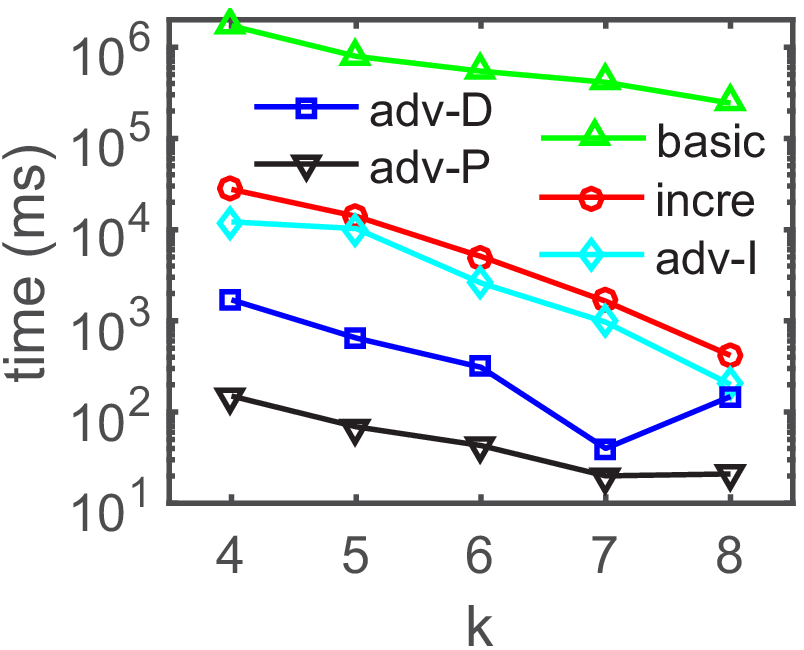}
  \end{minipage}
  &
  \begin{minipage}{4.175cm}
  \includegraphics[width=4.175cm]{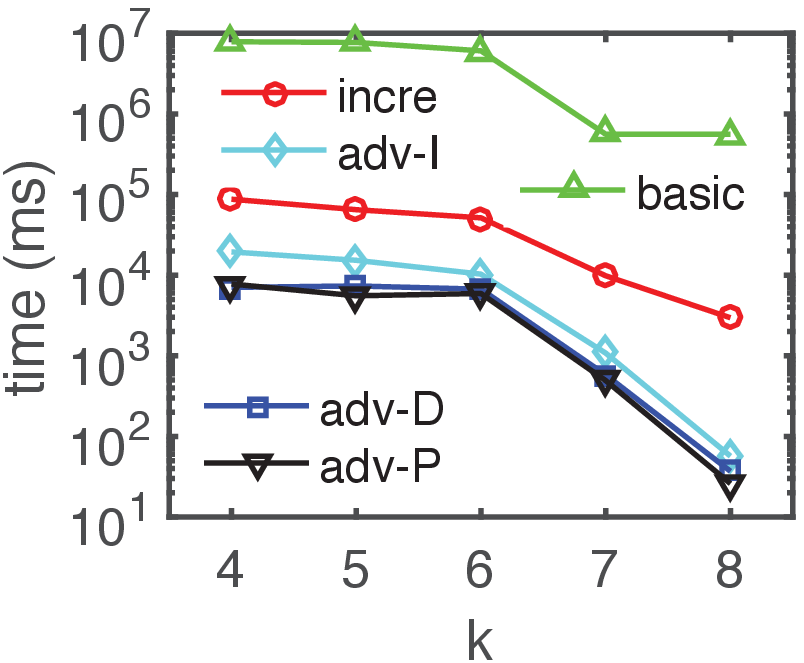}
  \end{minipage}
  \\
    (a) ACMDL (efficiency)
  &
   (b) Flickr (efficiency)
  &
   (c) PubMed (efficiency)
  &
   (d) DBLP (efficiency)
  \\
\begin{minipage}{4.175cm}
  \includegraphics[width=4.175cm]{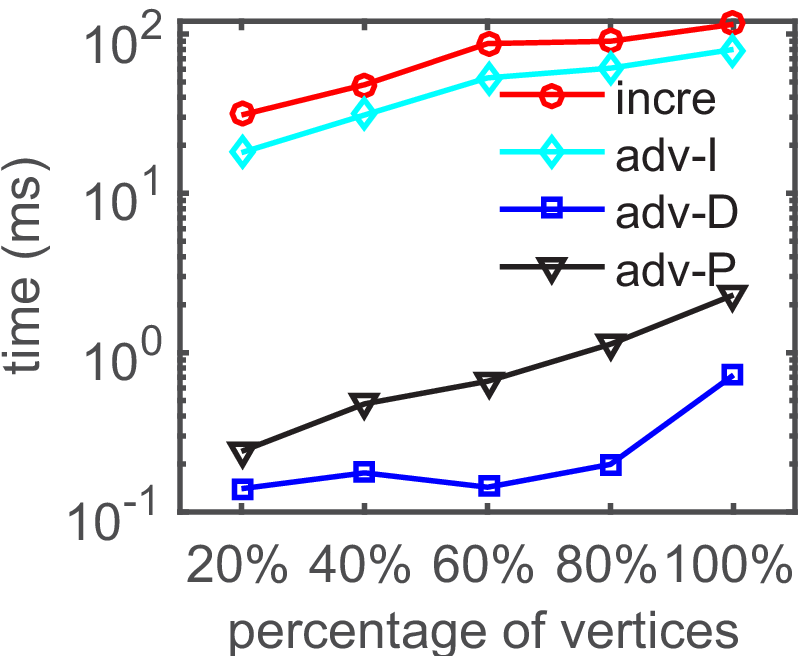}
  \end{minipage}
  &
  \begin{minipage}{4.175cm}
  \includegraphics[width=4.175cm]{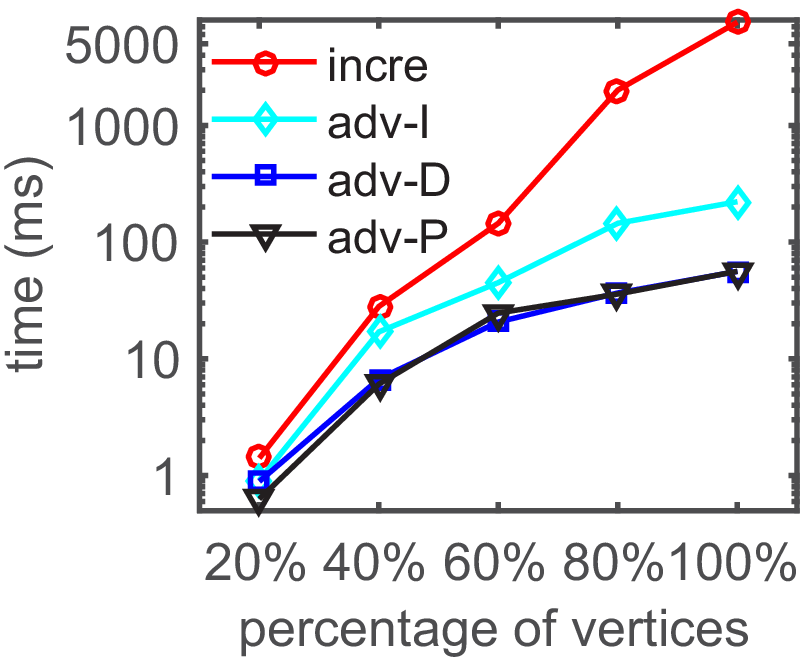}
  \end{minipage}
  &
   \begin{minipage}{4.175cm}
  \includegraphics[width=4.175cm]{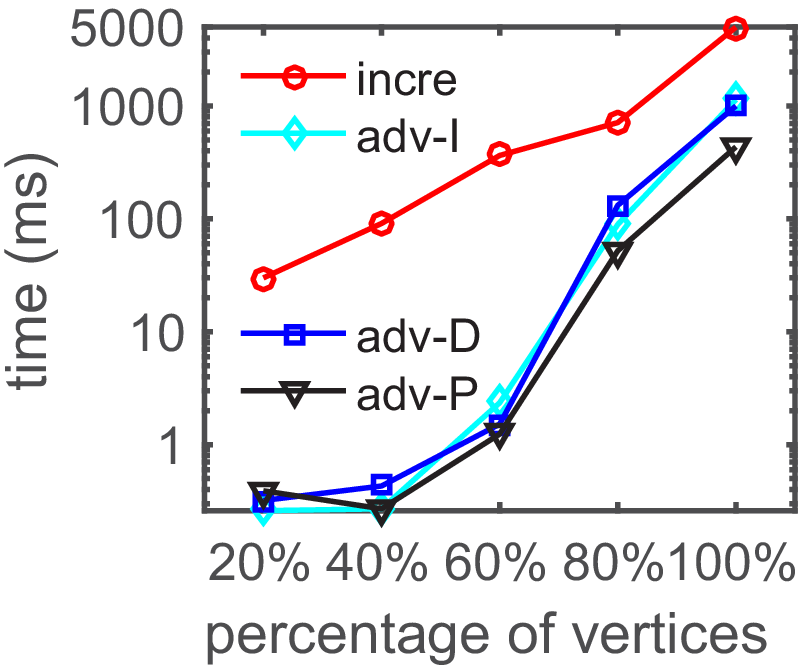}
  \end{minipage}
  &
  \begin{minipage}{4.175cm}
  \includegraphics[width=4.175cm]{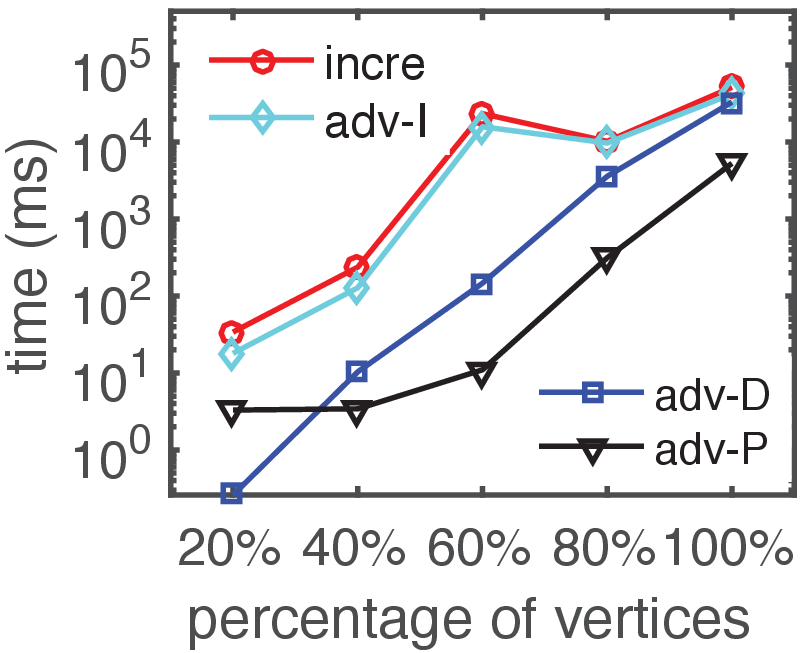}
  \end{minipage}
  \\
   (e) ACMDL (vertex scalab.)
  &
     (f) Flickr (vertex scalab.)
  &
     (g) pubMed (vertex scalab.)
   &
     (h) DBLP (vertex scalab.)
   \\
\begin{minipage}{4.175cm}
  \includegraphics[width=4.175cm]{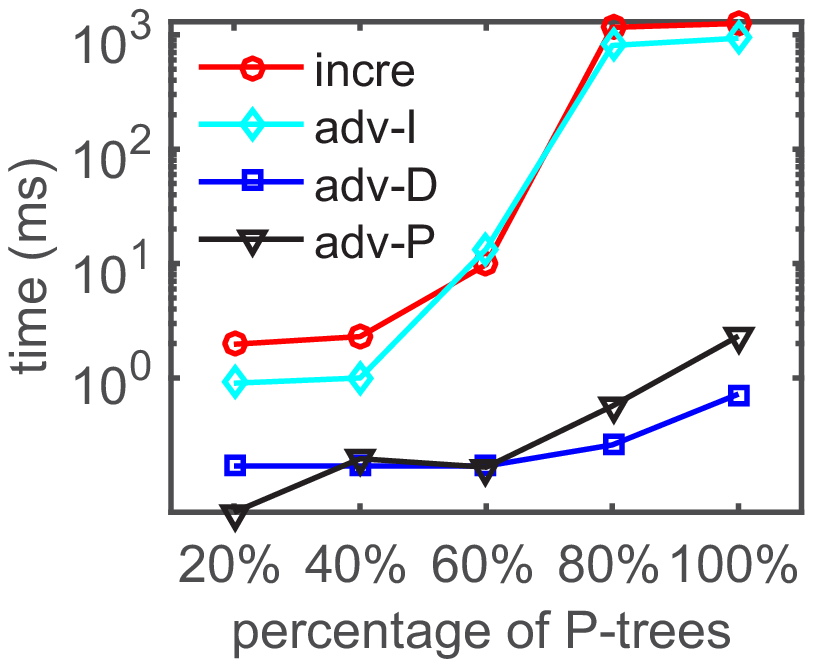}
  \end{minipage}
  &
  \begin{minipage}{4.175cm}
  \includegraphics[width=4.175cm]{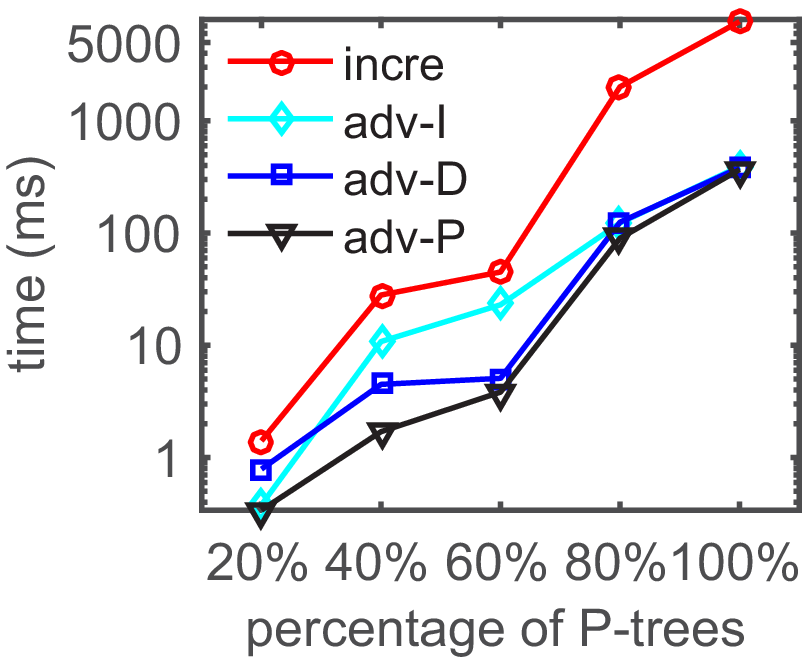}
  \end{minipage}
  &
  \begin{minipage}{4.175cm}
  \includegraphics[width=4.175cm]{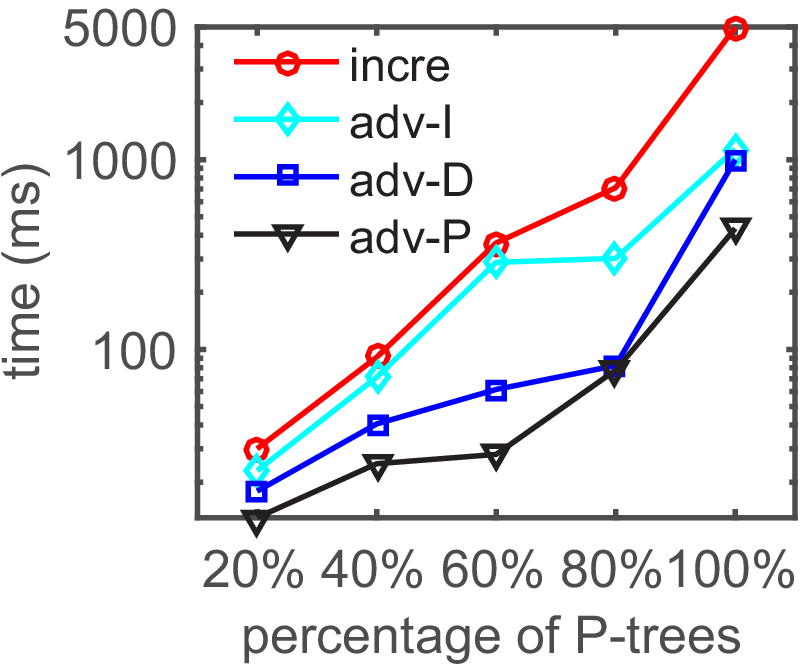}
  \end{minipage}
  &
  \begin{minipage}{4.175cm}
  \includegraphics[width=4.175cm]{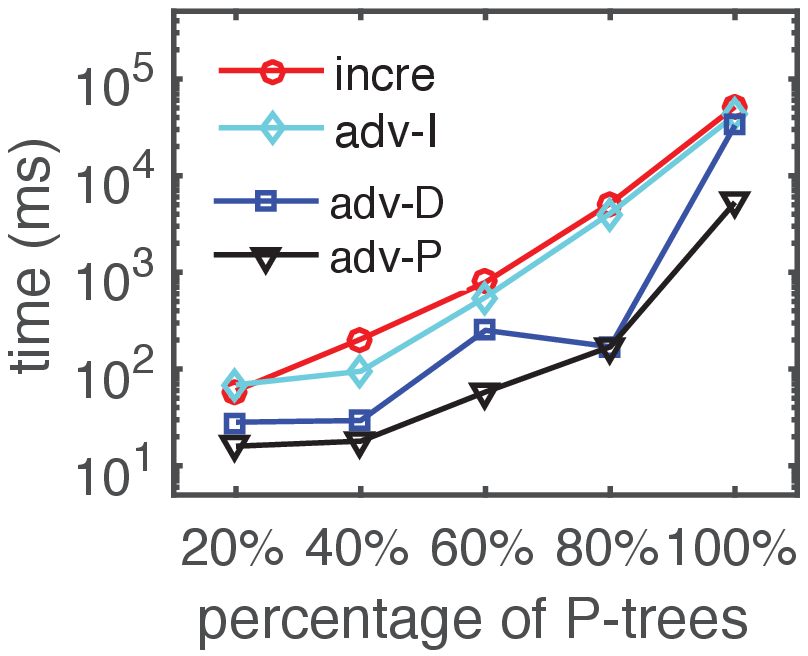}
  \end{minipage}
  \\
   (i) ACMDL (P-tree scalab.)
  &
     (j) Flickr (P-tree scalab.)
  &
     (k) pubMed (P-tree scalab.)
   &
     (l) DBLP (P-tree scalab.)
   \\
  \begin{minipage}{4.175cm}
  \includegraphics[width=4.175cm]{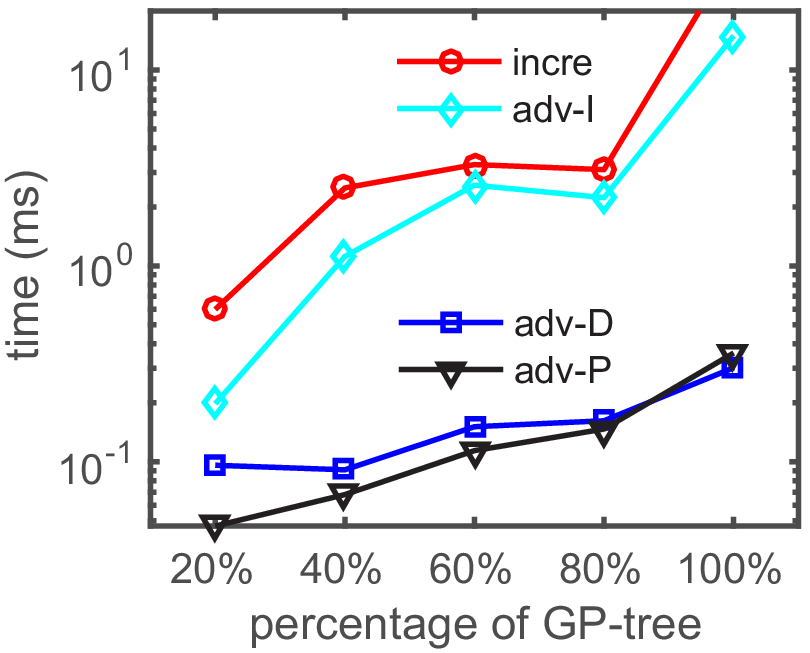}
  \end{minipage}
  &
  \begin{minipage}{4.175cm}
  \includegraphics[width=4.175cm]{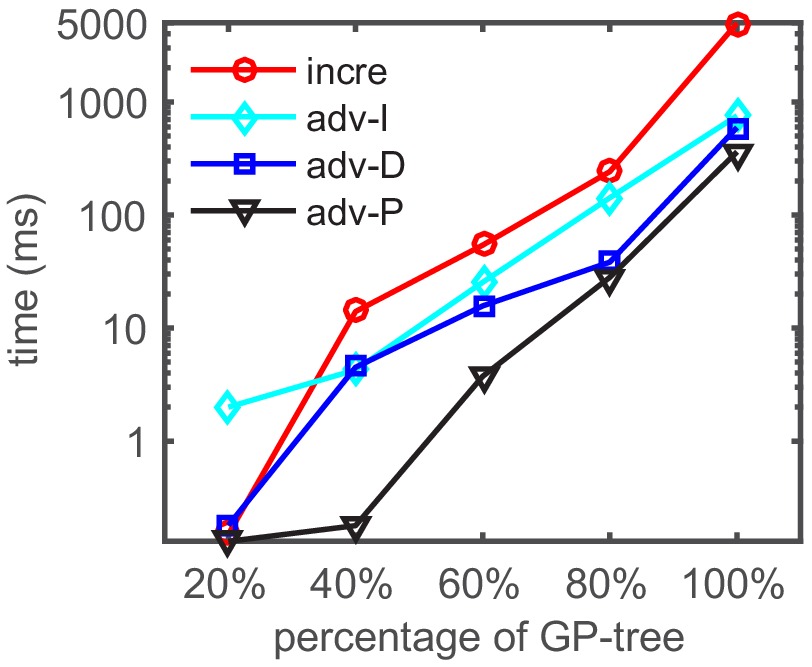}
  \end{minipage}
  &
  \begin{minipage}{4.175cm}
  \includegraphics[width=4.175cm]{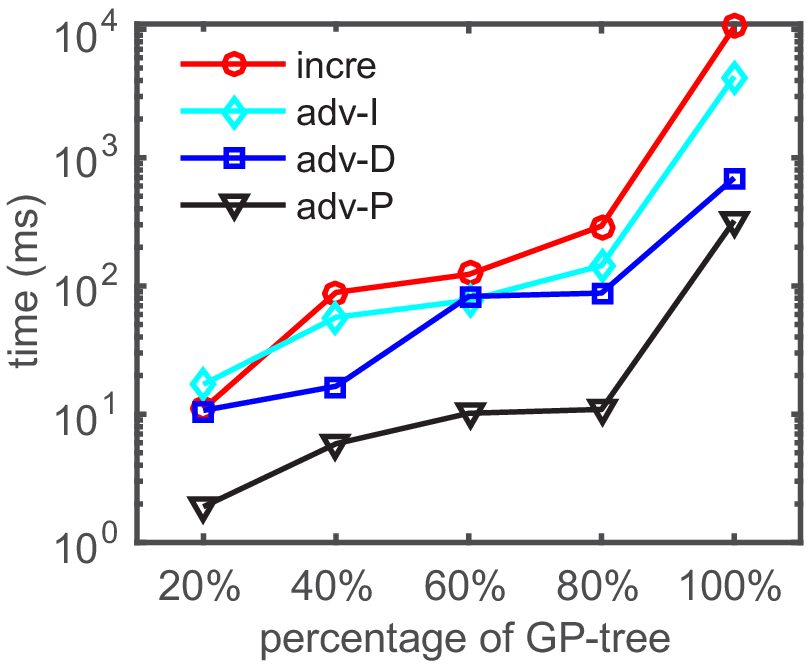}
  \end{minipage}
  &
   \begin{minipage}{4.175cm}
  \includegraphics[width=4.175cm]{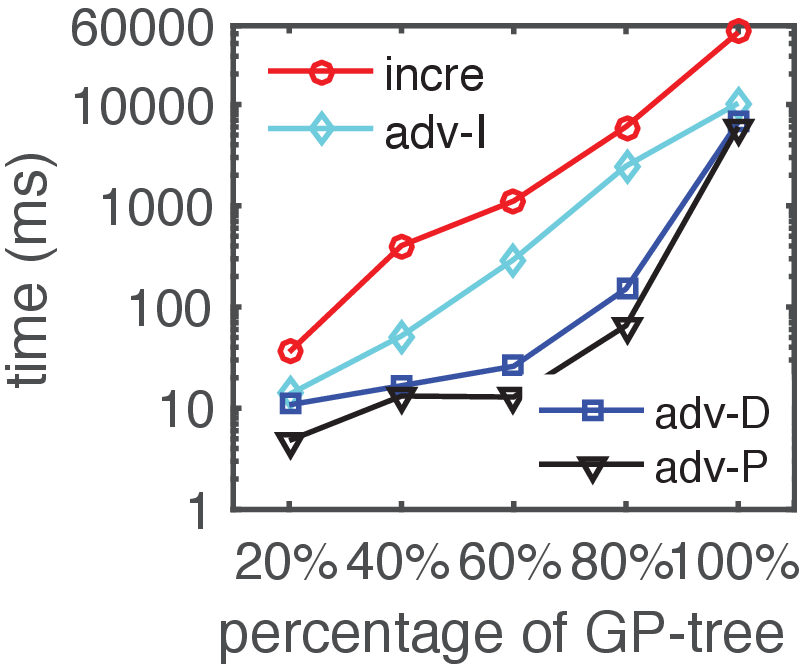}
  \end{minipage}
  \\
     (m) ACMDL (GP-tree scalab.)
  &
     (n) Flickr (GP-tree scalab.)
  &
     (o) pubMed (GP-tree scalab.)
   &
     (p) DBLP (GP-tree scalab.)
   \\ 
   \begin{minipage}{4.175cm}
  \includegraphics[width=4.175cm]{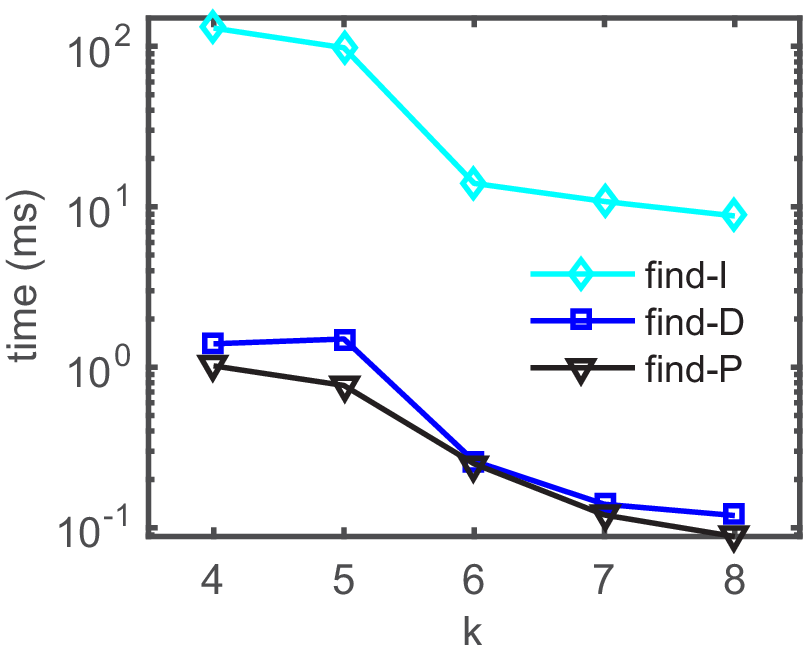}
  \end{minipage}
  &
  \begin{minipage}{4.175cm}
  \includegraphics[width=4.175cm]{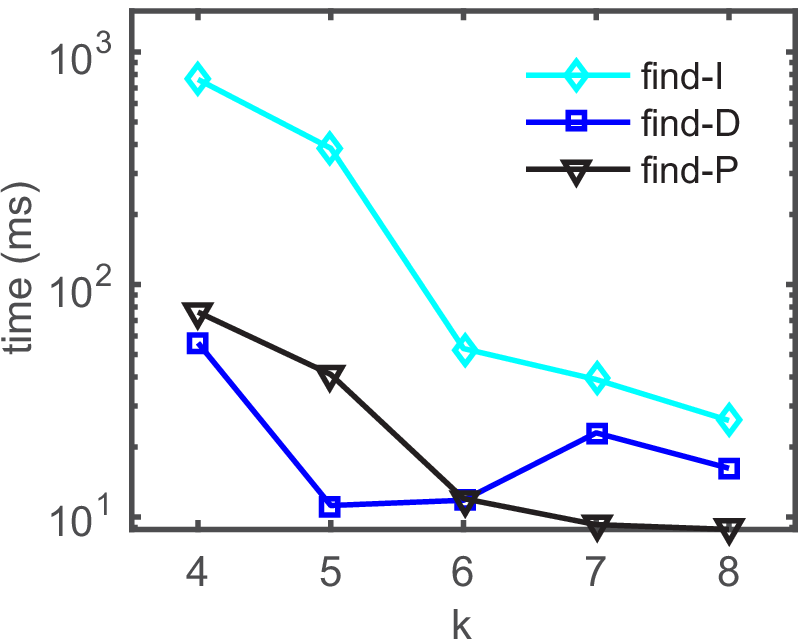}
  \end{minipage}
  &
  \begin{minipage}{4.175cm}
  \includegraphics[width=4.175cm]{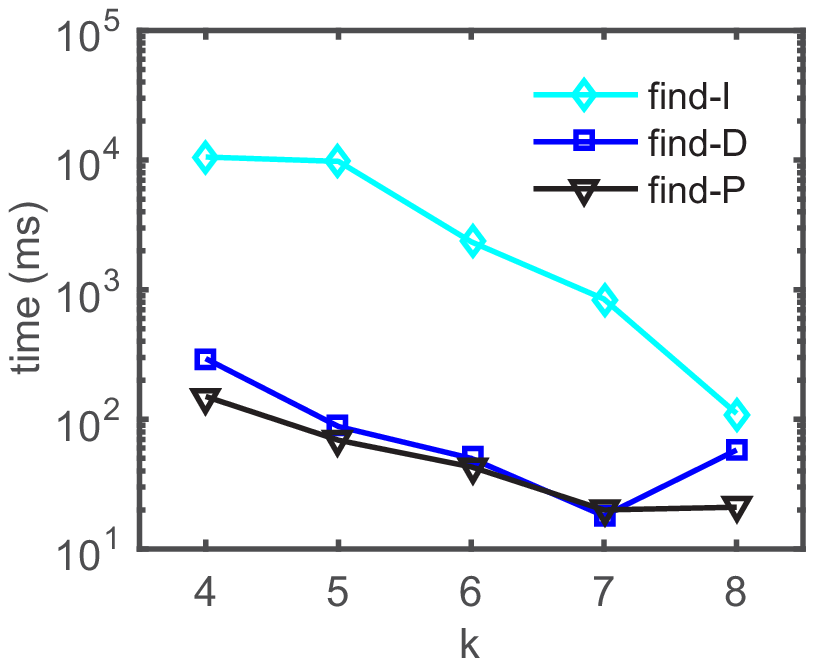}
  \end{minipage}
  &
  \begin{minipage}{4.175cm}
  \includegraphics[width=4.175cm]{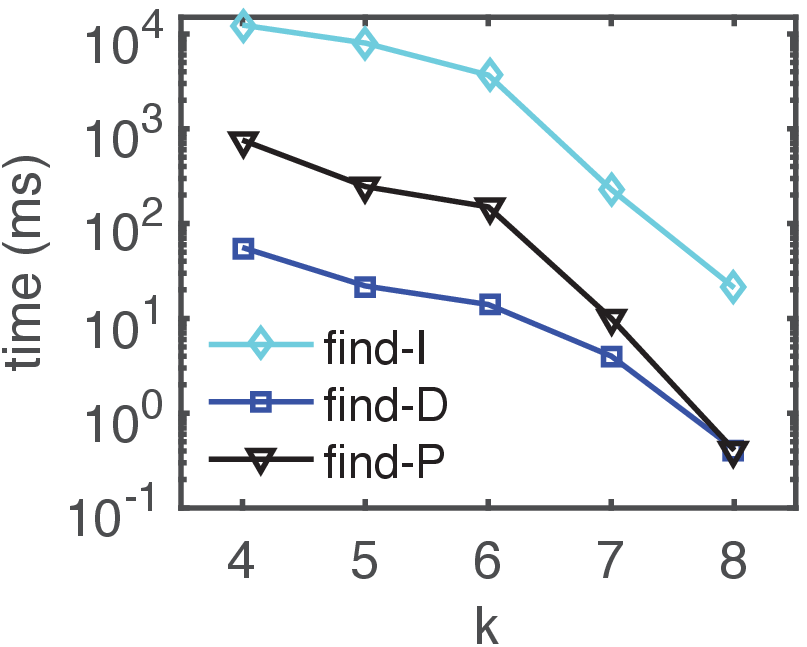}
  \end{minipage}
  \\
    (q) ACMDL ({\tt find} functions)
  &
     (r) Flickr ({\tt find} functions)
  &
     (s) pubMed ({\tt find} functions)
   &
     (t) DBLP ({\tt find} functions)
\end{tabular}
\caption{Efficiency and scalability on ACMDL, Flickr, PubMed and DBLP datasets.}
\label{fig:exp-cs-tkde2}
\end{figure*}

\section{Conclusions and Future Directions}
\label{sec:conclusion}

In this paper,we study the online community search problem which exhibit both semantic and strucutre cohesiveness on large-scale profiled graphs. 
Given a vertex $q$ of a profiled graph $G$, we study the PCS problem, which aims to find profiled communities containing $q$. We firstly introduce a basic solution. To further accelarate the query efficiency, we develop an index and some index-based query algorithms accordingly. We evalute the algorithms on both real and synthetic datasets, and our experimental results demonstrate the effectiveness of PCS and the efficiency of our solutions. 

In the future, we will study other structure cohesiveness measures (e.g., $k$-truss and $k$-clique) and profile cohesiveness measures in the PCS definition. An potential extension of PCS is to explore how to relax the query condition and optimize our solutions. For instance, we can stimulate that each vertex of the targeted community has a semantic similarity with the query vertex $q$ to be at least $\beta$ ($\beta$ ≥ 0), where $\beta$ is a predefined threshold. Or we can try to relax the structure cohesiveness (e.g., the proportion of vertices in a community having degrees of at least $k$ is at least $\delta$ where parameter $\delta$ > 0). Another useful question is to examine how the directions of edges will affect the formation of an PC. For example, D-core [14], a concept extended from k-core for directed graphs, can be utilized to measure the structure cohesiveness, and develop algorithms that is similar to those of PCS. In practice, directed edges are more common in many real social networks. In addition, we will also investigate other problems that can be done on profiled graphs, such as communities search family of problems, or labeled graph search problems. We will study how other graph pattern matching techniques can be extended to find PCs on profiled graphs and how to automatically generate a meaningful graph pattern that well reflects the ground-truth communities.

\vspace{-0.1in}
\section*{Acknowledgment}
Reynold Cheng, Yixiang Fang, and Yankai Chen were supported by the Research Grants Council of Hong Kong (RGC Projects HKU 106150091, 17229116, 17205115) and the University of Hong Kong (Projects 104004572, 102009508 ,104004129).
Xiaojun Chen was supported by NSFC under Grant no. 61773268. Jie Zhang was supported by the MOE AcRF Tier 1 funding (M4011894.020) and the Telenor-NTU Joint R\&D funding. We thank the editors and reviewers for their insightful comments.

\small
\bibliographystyle{IEEEtran}
\bibliography{ref}

\begin{thebibliography}{10}
\providecommand{\url}[1]{#1}
\csname url@samestyle\endcsname
\providecommand{\newblock}{\relax}
\providecommand{\bibinfo}[2]{#2}
\providecommand{\BIBentrySTDinterwordspacing}{\spaceskip=0pt\relax}
\providecommand{\BIBentryALTinterwordstretchfactor}{4}
\providecommand{\BIBentryALTinterwordspacing}{\spaceskip=\fontdimen2\font plus
\BIBentryALTinterwordstretchfactor\fontdimen3\font minus
  \fontdimen4\font\relax}
\providecommand{\BIBforeignlanguage}[2]{{%
\expandafter\ifx\csname l@#1\endcsname\relax
\typeout{** WARNING: IEEEtran.bst: No hyphenation pattern has been}%
\typeout{** loaded for the language `#1'. Using the pattern for}%
\typeout{** the default language instead.}%
\else
\language=\csname l@#1\endcsname
\fi
#2}}
\providecommand{\BIBdecl}{\relax}
\BIBdecl

\bibitem{ding2007finding}
B.~Ding, J.~X. Yu, S.~Wang, L.~Qin, X.~Zhang, and X.~Lin, ``Finding top-k
  min-cost connected trees in databases,'' in \emph{Data Engineering, 2007.
  ICDE 2007. IEEE 23rd International Conference on}.\hskip 1em plus 0.5em minus
  0.4em\relax IEEE, 2007, pp. 836--845.

\bibitem{fang2014detecting}
Y.~Fang, H.~Zhang, Y.~Ye, and X.~Li, ``Detecting hot topics from twitter: A
  multiview approach,'' \emph{Journal of Information Science}, vol.~40, no.~5,
  pp. 578--593, 2014.

\bibitem{he2007blinks}
H.~He, H.~Wang, J.~Yang, and P.~S. Yu, ``Blinks: ranked keyword searches on
  graphs,'' in \emph{Proceedings of the 2007 ACM SIGMOD international
  conference on Management of data}.\hskip 1em plus 0.5em minus 0.4em\relax
  ACM, 2007, pp. 305--316.

\bibitem{kacholia2005bidirectional}
V.~Kacholia, S.~Pandit, S.~Chakrabarti, S.~Sudarshan, R.~Desai, and
  H.~Karambelkar, ``Bidirectional expansion for keyword search on graph
  databases,'' in \emph{Proceedings of the 31st international conference on
  Very large data bases}.\hskip 1em plus 0.5em minus 0.4em\relax VLDB
  Endowment, 2005, pp. 505--516.

\bibitem{kargar2011keyword}
M.~Kargar and A.~An, ``Keyword search in graphs: Finding r-cliques,''
  \emph{Proceedings of the VLDB Endowment}, vol.~4, no.~10, pp. 681--692, 2011.

\bibitem{xu2012model}
Z.~Xu, Y.~Ke, Y.~Wang, H.~Cheng, and J.~Cheng, ``A model-based approach to
  attributed graph clustering,'' in \emph{Proceedings of the 2012 ACM SIGMOD
  international conference on management of data}.\hskip 1em plus 0.5em minus
  0.4em\relax ACM, 2012, pp. 505--516.

\bibitem{yu2009keyword}
J.~X. Yu, L.~Qin, and L.~Chang, ``Keyword search in databases,''
  \emph{Synthesis Lectures on Data Management}, vol.~1, no.~1, pp. 1--155,
  2009.

\bibitem{sozio2010community}
M.~Sozio and A.~Gionis, ``The community-search problem and how to plan a
  successful cocktail party,'' in \emph{KDD}, 2010, pp. 939--948.

\bibitem{online-sigmod2013}
W.~Cui, Y.~Xiao, H.~Wang, Y.~Lu, and W.~Wang, ``Online search of overlapping
  communities,'' in \emph{SIGMOD}, 2013, pp. 277--288.

\bibitem{k-truss2014}
X.~Huang, H.~Cheng, L.~Qin, W.~Tian, and J.~X. Yu, ``Querying k-truss community
  in large and dynamic graphs,'' in \emph{SIGMOD}, 2014.

\bibitem{fang2016effective}
Y.~Fang, R.~Cheng, S.~Luo, and J.~Hu, ``Effective community search for large
  attributed graphs,'' \emph{PVLDB}, vol.~9, no.~12, pp. 1233--1244, 2016.

\bibitem{huang2017attribute}
X.~Huang and L.~V. Lakshmanan, ``Attribute-driven community search,''
  \emph{PVLDB}, vol.~10, no.~9, pp. 949--960, 2017.

\bibitem{fortunato2010community}
S.~Fortunato, ``Community detection in graphs,'' \emph{Physics reports}, vol.
  486, no. 3-5, pp. 75--174, 2010.

\bibitem{liu2009topic}
Y.~Liu, A.~Niculescu-Mizil, and W.~Gryc, ``Topic-link lda: joint models of
  topic and author community,'' in \emph{proceedings of the 26th annual
  international conference on machine learning}.\hskip 1em plus 0.5em minus
  0.4em\relax ACM, 2009, pp. 665--672.

\bibitem{attr-topic-kdd2008}
R.~M. Nallapati, A.~Ahmed, E.~P. Xing, and W.~W. Cohen, ``Joint latent topic
  models for text and citations,'' in \emph{KDD}, 2008.

\bibitem{newman2004finding}
M.~E. Newman and M.~Girvan, ``Finding and evaluating community structure in
  networks,'' \emph{Physical review E}, vol.~69, no.~2, p. 026113, 2004.

\bibitem{attr-www2013}
Y.~Ruan, D.~Fuhry, and S.~Parthasarathy, ``Efficient community detection in
  large networks using content and links,'' in \emph{WWW}, 2013.

\bibitem{attr-topic-sigmod2012}
Z.~Xu, Y.~Ke, Y.~Wang, H.~Cheng, and J.~Cheng, ``A model-based approach to
  attributed graph clustering,'' in \emph{SIGMOD}, 2012.

\bibitem{yang2013community}
J.~Yang, J.~McAuley, and J.~Leskovec, ``Community detection in networks with
  node attributes,'' in \emph{ICDM}, 2013, pp. 1151--1156.

\bibitem{zhou2009graph}
Y.~Zhou, H.~Cheng, and J.~X. Yu, ``Graph clustering based on
  structural/attribute similarities,'' \emph{Proceedings of the VLDB
  Endowment}, vol.~2, no.~1, pp. 718--729, 2009.

\bibitem{barbieri2015efficient}
N.~Barbieri, F.~Bonchi, E.~Galimberti, and F.~Gullo, ``Efficient and effective
  community search,'' \emph{Data mining and knowledge discovery}, vol.~29,
  no.~5, pp. 1406--1433, 2015.

\bibitem{cui2013online}
W.~Cui, Y.~Xiao, H.~Wang, Y.~Lu, and W.~Wang, ``Online search of overlapping
  communities,'' in \emph{SIGMOD}, 2013, pp. 277--288.

\bibitem{local2014}
W.~Cui, Y.~Xiao, H.~Wang, and W.~Wang, ``Local search of communities in large
  graphs,'' in \emph{SIGMOD}, 2014, pp. 991--1002.

\bibitem{huang2015approximate}
X.~Huang, L.~V. Lakshmanan, J.~X. Yu, and H.~Cheng, ``Approximate closest
  community search in networks,'' \emph{Proceedings of the VLDB Endowment},
  vol.~9, no.~4, pp. 276--287, 2015.

\bibitem{cui2014local}
W.~Cui, Y.~Xiao, H.~Wang, and W.~Wang, ``Local search of communities in large
  graphs,'' in \emph{SIGMOD}, 2014, pp. 991--1002.

\bibitem{li2015influential}
R.-H. Li, L.~Qin, J.~X. Yu, and R.~Mao, ``Influential community search in large
  networks,'' \emph{PVLDB}, vol.~8, no.~5, pp. 509--520, 2015.

\bibitem{batagelj2003m}
V.~Batagelj and M.~Zaversnik, ``An o (m) algorithm for cores decomposition of
  networks,'' \emph{arXiv preprint cs/0310049}, 2003.

\bibitem{dorogovtsev2006k}
S.~N. Dorogovtsev, A.~V. Goltsev, and J.~F.~F. Mendes, ``K-core organization of
  complex networks,'' \emph{Physical review letters}, vol.~96, no.~4, p.
  040601, 2006.

\bibitem{seidman1983network}
S.~B. Seidman, ``Network structure and minimum degree,'' \emph{Social
  networks}, vol.~5, no.~3, pp. 269--287, 1983.

\bibitem{guo2008regionalization}
D.~Guo, ``Regionalization with dynamically constrained agglomerative clustering
  and partitioning (redcap),'' \emph{International Journal of Geographical
  Information Science}, vol.~22, no.~7, pp. 801--823, 2008.

\bibitem{expert2011uncovering}
P.~Expert, T.~S. Evans, V.~D. Blondel, and R.~Lambiotte, ``Uncovering
  space-independent communities in spatial networks,'' \emph{Proceedings of the
  National Academy of Sciences}, vol. 108, no.~19, pp. 7663--7668, 2011.

\bibitem{chen2015finding}
Y.~Chen, J.~Xu, and M.~Xu, ``Finding community structure in spatially
  constrained complex networks,'' \emph{International Journal of Geographical
  Information Science}, vol.~29, no.~6, pp. 889--911, 2015.

\bibitem{newman2006modularity}
M.~E. Newman, ``Modularity and community structure in networks,''
  \emph{Proceedings of the national academy of sciences}, vol. 103, no.~23, pp.
  8577--8582, 2006.

\bibitem{qi2013online}
G.-J. Qi, C.~C. Aggarwal, and T.~S. Huang, ``Online community detection in
  social sensing,'' in \emph{Proceedings of the sixth ACM international
  conference on Web search and data mining}.\hskip 1em plus 0.5em minus
  0.4em\relax ACM, 2013, pp. 617--626.

\bibitem{ruan2013efficient}
Y.~Ruan, D.~Fuhry, and S.~Parthasarathy, ``Efficient community detection in
  large networks using content and links,'' in \emph{Proceedings of the 22nd
  international conference on World Wide Web}.\hskip 1em plus 0.5em minus
  0.4em\relax ACM, 2013, pp. 1089--1098.

\bibitem{attr-topic-icml2009}
Y.~Liu, A.~Niculescu-Mizil, and W.~Gryc, ``Topic-link lda: joint models of
  topic and author community,'' in \emph{ICML}, 2009.

\bibitem{nallapati2008joint}
R.~M. Nallapati, A.~Ahmed, E.~P. Xing, and W.~W. Cohen, ``Joint latent topic
  models for text and citations,'' in \emph{Proceedings of the 14th ACM SIGKDD
  international conference on Knowledge discovery and data mining}.\hskip 1em
  plus 0.5em minus 0.4em\relax ACM, 2008, pp. 542--550.

\bibitem{sachan2012using}
M.~Sachan, D.~Contractor, T.~A. Faruquie, and L.~V. Subramaniam, ``Using
  content and interactions for discovering communities in social networks,'' in
  \emph{Proceedings of the 21st international conference on World Wide
  Web}.\hskip 1em plus 0.5em minus 0.4em\relax ACM, 2012, pp. 331--340.

\bibitem{hu2016querying}
J.~Hu, X.~Wu, R.~Cheng, S.~Luo, and Y.~Fang, ``Querying minimal steiner
  maximum-connected subgraphs in large graphs,'' in \emph{Proceedings of the
  25th ACM International on Conference on Information and Knowledge
  Management}.\hskip 1em plus 0.5em minus 0.4em\relax ACM, 2016, pp.
  1241--1250.

\bibitem{fang2017effective}
Y.~Fang, R.~Cheng, Y.~Chen, S.~Luo, and J.~Hu, ``Effective and efficient
  attributed community search,'' \emph{VLDBJ}, pp. 1--26, 2017.

\bibitem{md1983}
S.~B. Seidman, ``Network structure and minimum degree,'' \emph{Social
  networks}, vol.~5, no.~3, pp. 269--287, 1983.

\bibitem{asai2004efficient}
T.~Asai, K.~Abe, S.~Kawasoe, H.~Sakamoto, H.~Arimura, and S.~Arikawa,
  ``Efficient substructure discovery from large semi-structured data,''
  \emph{IEICE TRANSACTIONS on Information and Systems}, vol.~87, no.~12, pp.
  2754--2763, 2004.

\bibitem{thomas2010margin}
L.~T. Thomas, S.~R. Valluri, and K.~Karlapalem, ``Margin: Maximal frequent
  subgraph mining,'' \emph{ACM TKDD}, vol.~4, no.~3, p.~10, 2010.

\bibitem{thomee2015new}
B.~Thomee~et al, ``The new data and new challenges in multimedia research,''
  \emph{arXiv:1503.01817}, 2015.

\bibitem{leskovec2012learning}
J.~Leskovec and J.~J. Mcauley, ``Learning to discover social circles in ego
  networks,'' in \emph{NIPS}, 2012, pp. 539--547.

\end{thebibliography}

{\sffamily
\noindent \textbf {Yankai Chen}
received the M.sc degree from the University of Hong Kong in 2018 and B.Sc degree from Nanjing University in 2016. 
Currently, he is a Ph.D. student in the School of Computer Science and Engineering, Nanyang Technological University (NTU), under the supervision of A/Prof. Jie Zhang NTU. His research area is machine learning and artificial intelligence, especially in graph data mining and recommendation system.  

\noindent \textbf {Yixiang Fang} received the Ph.D. degree from the University of Hong Kong in 2017. Currently, he is a postdoc researcher in University of New South Wales, under the supervision of Prof. Xuemin Lin. His research interests focus on graph queries.

\noindent \textbf {Reynold Cheng} is an associate professor in the Department of Computer Science, University of Hong Kong. He has served as a PC member and reviewer for international conferences (e.g., SIGMOD, VLDB, ICDE and KDD) and journals (e.g., TKDE, TODS, VLDBJ and IS). Currently, he is an associate editor of TKDE journal, and on the EIC selection committee for TKDE. He is also a member of the IEEE and ACM.

\noindent \textbf {Yun Li} currently is a third year graduate student of Department of Computer Science and Technology in Nanjing University and a member of IIP Group, led by professor Chongjun Wang. She received the B.Sc degree from Nanjing University in 2016. Recently her research interests focus on machine learning, natural language process, spatial database and graph Database.

\noindent \textbf {Xiaojun Chen} is an assistant professor in Shenzhen University. He received the PhD degree from the Harbin Institute of Technology in 2011. His research interests include clustering and massive data mining. He is a member of the IEEE.

\noindent \textbf {Jie Zhang} is an Associate Professor of the School of Computer Science and Engineering, Nanyang Technological University. He obtained Ph.D. in Cheriton School of Computer Science from University of Waterloo in 2009. His research is in the general area of Artificial Intelligence and Multi-Agent Systems.
}

\end{document}